\documentclass[11pt]{article}

\usepackage{float, amssymb, amsmath, amsthm, graphicx, bbm, multirow, siunitx, placeins, authblk, geometry, fullpage}
\usepackage[round]{natbib}

\usepackage[dvipsnames]{xcolor}
\usepackage{diagbox}
\usepackage{tikz}
\usetikzlibrary{arrows.meta}
\usetikzlibrary{decorations.pathreplacing}
\usepackage[shortlabels]{enumitem}
\setlist{nolistsep}

\usepackage{array}
\newcolumntype{C}[1]{>{\centering\arraybackslash}p{#1}}

\usepackage{algorithm}
\usepackage{algpseudocode}
\algnewcommand\algorithmicinput{\textbf{INPUT:}}
\algnewcommand\INPUT{\item[\algorithmicinput]}

\algnewcommand\algorithmicoutput{\textbf{OUTPUT:}}
\algnewcommand\OUTPUT{\item[\algorithmicoutput]}

\algnewcommand\algorithmicexploration{\textbf{Initialise:}}
\algnewcommand\Initialise{\item[\algorithmicexploration]}

\algnewcommand\algorithmicexploitation{\textbf{Exploitation:}}
\algnewcommand\Exploitation{\item[\algorithmicexploitation]}

\usepackage{xr-hyper}
\usepackage{hyperref}[]
\hypersetup{
    colorlinks=true,
    linkcolor=blue,
    filecolor=magenta,      
    urlcolor=cyan,
      citecolor=blue,
    }

\usepackage{cleveref}
\usepackage{booktabs}

\newtheorem{theorem}{Theorem}
\newtheorem{lemma}[theorem]{Lemma}
\newtheorem{proposition}[theorem]{Proposition}

\newtheorem{definition}{Definition}

\newtheorem{assumption}{Assumption}
\newtheorem{model}{Model}

\makeatletter
\newcommand{\neutralize}[1]{\expandafter\let\csname c@#1\endcsname\count@}
\makeatother

\allowdisplaybreaks

\DeclareMathOperator*{\argmin}{argmin}

\DeclareMathOperator*{\argmax}{argmax}

\usepackage{mathtools}

\def\E{\mathbb{E}}

\def\R{\mathbb{R}}

\title{Change Point Localization and Inference\\ in Dynamic Multilayer Networks}

\author[1]{Fan Wang}
\author[2]{Kyle Ritscher}
\author[3]{Yik Lun Kei}
\author[4]{Xin Ma}
\author[2]{Oscar Hernan Madrid Padilla}
\affil[1]{Department of Statistics,
  University of Warwick}
\affil[2]{Department of Statistics,
  University of California, Los Angeles}
\affil[3]{Department of Statistics,
  University of California, Santa Cruz}
\affil[4]{Department of Biostatistics, Columbia University}

\begin{document}

\maketitle

\begin{abstract}
We study offline change point localization and inference in dynamic multilayer random dot product graphs (D-MRDPGs), where at each time point, a multilayer network is observed with shared node latent positions and time-varying, layer-specific connectivity patterns. 
We propose a novel two-stage algorithm that combines seeded binary segmentation with low-rank tensor estimation, and establish its consistency in estimating both the number and locations of change points.
Furthermore, we derive the limiting distributions of the refined estimators under both vanishing and non-vanishing jump regimes. To the best of our knowledge, this is the first result of its kind in the context of dynamic network data.
We also develop a fully data-driven procedure for constructing confidence intervals. Extensive numerical experiments demonstrate the superior performance and practical utility of our methods compared to existing alternatives.
\end{abstract}

\section{Introduction}\label{section:intro}
Statistical network analysis studies relationships among entities represented as nodes connected by edges.
While single-layer networks capture pairwise relationships efficiently, many real-world systems involve multiple types of interaction among the same set of nodes. 
Multilayer networks address this complexity by organizing these varied interactions into distinct layers over a common node set.
For instance, in social networks, individuals may simultaneously interact through various relationships such as collaboration and friendship, each forming a distinct layer \cite[e.g.][]{Porter2018WhatIA}. Modeling multilayer structures enables the capture of heterogeneity and the identification of common structures across various interaction types. 

In practice, network structures often evolve over time. For instance, transportation networks may exhibit gradual diurnal variations or sudden structural changes due to unexpected events like accidents or road closures. Detecting these sudden shifts and providing adaptive strategies, such as dynamic traffic signal control or rerouting recommendations, is crucial for efficient transportation management. These abrupt structural shifts are referred to as change points.  This naturally falls in the territory of change point analysis.

Change point analysis is a well-established area in statistics concerned with detecting abrupt structural changes in ordered data. It can be broadly classified into online and offline settings, depending on whether data are analyzed sequentially as they are collected or retrospectively after the full dataset has been observed. In the context of dynamic networks, online change point detection has been studied in models such as inhomogeneous Bernoulli networks \citep[e.g.,][]{yu2021optimal} and random weighted edge networks \citep[e.g.,][]{chen2024monitoring}. Offline detection has been explored in various network models, including inhomogeneous Bernoulli networks \citep[e.g.,][]{wang2021optimal}, stochastic block models \citep[e.g.,][]{xu2022statistical, bhattacharjee2020change} and random dot product graphs \citep[e.g.,][]{padilla2022change}. More recently, \cite{wang2023multilayer} investigated online change point detection in dynamic multilayer random dot product graphs (D-MRDPGs).

In this paper, we study offline change point localization and inference for D-MRDPGs.  Specifically, at each time point, we observe a realization of an $L$-layered multilayer network, where nodes are associated with fixed but latent positions, and layer-specific weight matrices capture heterogeneous interactions across layers. These weight matrices are allowed to vary over time. Our goal is to develop efficient procedures for localizing and inferring change points under this dynamic multilayer structure in the offline setting.

\subsection{List of contributions}

The main contributions of this paper are summarized as follows. 

First, to the best of our knowledge, this is the first work to address offline change point detection in dynamic multilayer networks.  We propose a novel two-stage procedure: the first stage generates a coarse set of candidates using seeded binary segmentation, in the spirit of \cite{kovacs2023seeded}, combined with CUSUM statistics. The second stage refines these candidates utilizing low-rank tensor estimation techniques. Under appropriate conditions, we establish consistency in both the number and locations of estimated change points. 

Second, we derive the limit distributions of the refined change point estimators.  Depending on whether the jump size remains fixed or vanishes as the time horizon diverges, the limiting distributions exhibit two distinct regimes. To the best of our knowledge, these are the first such results established in the network literature.  We further develop a completely data-driven procedure for constructing confidence intervals for the true change points.

Lastly, we conduct extensive numerical experiments to assess the performance of our proposed methods, demonstrating substantial improvements over existing state-of-the-art algorithms.

\subsection{Notation and organization}\label{sec:notation}

For any positive integer $p$, let $[p] = \{1, \ldots, p\}$. 
Let $\{a_n\}_{n \in \mathbb{N}^+}$ and $\{b_n\}_{n \in \mathbb{N}^+}$ be sequences of positive real numbers.  We write $a_n = O(b_n)$ if
$a_n \leq C b_n$ for some constant $C > 0$ independent of $n$ and all sufficiently large $n$, and  
$a_n = \Theta(b_n)$ if both $a_n = O(b_n)$ and $b_n = O(a_n)$. 
For a sequence of random variables $\{X_n\}_{n \in \mathbb{N}^+}$, we write  $ X_n = O_p(a_n) $ if $ \lim_{M \to \infty} \limsup_n \mathbb{P}(\vert X_n \vert  \geq M a_n) = 0 $. For any two sets $\mathcal{C}$ and $\mathcal{C}'$, define the one-sided Hausdorff distance $d(\mathcal{C}'|\mathcal{C}) = \max_{c \in \mathcal{C}} \min_{c' \in \mathcal{C}'} |c' - c|$ with  the convention that $d(\mathcal{C}'|\mathcal{C}) = \infty$ if either
$\mathcal{C}' = \emptyset$ or $\mathcal{C} = \emptyset$.

For any matrix $A \in \R^{p_1 \times p_2}$, let $A_i$ and $A^j$ be the $i$th row and $j$th column of $A$, respectively, and let $\sigma_1(A) \geq \dots \geq  \sigma_{p_1 \wedge p_2} (A)\geq 0$ be its singular values.  For any order-3 tensors $\mathbf{M}, \mathbf{Q}  \in \R^{p_1 \times p_2 \times p_3}$, define the inner product
$\langle \mathbf{M}, \mathbf{Q} \rangle = \sum_{i = 1}^{p_1} \sum_{j=1}^{p_2} \sum_{l = 1}^{p_3}   \mathbf{M} _{i, j, l } \mathbf{Q} _{i, j, l }$ and the Frobenius norm
$\|\mathbf{M}  \|_{\mathrm{F}}^2 =\langle \mathbf{M}, \mathbf{M} \rangle$.
The mode-$1$ matricization of a tensor  $\mathbf{M}$  is denoted by $\mathcal{M}_1(\mathbf{M}) \in \mathbb{R}^{p_{1} \times (p_{2} p_{3})}$
with entries $\mathcal{M}_1(\mathbf{M} )_{i_1, (i_2 - 1)p_{3} +i_3 } = \mathbf{M} _{i_1, i_2, i_3}$. The mode-2 and mode-3 matricizations are analogously defined as $\mathcal{M}_2(\mathbf{M}) \in \mathbb{R}^{p_2 \times (p_3 p_1)}$ and $\mathcal{M}_3(\mathbf{M}) \in \mathbb{R}^{p_3 \times (p_1 p_2)}$, respectively.
The Tucker ranks $(r_1, r_2, r_3)$ of $\mathbf{M}$ are given by $r_s = \mathrm{rank}(\mathcal{M}_s(\mathbf{M}))$ for $s \in [3]$. For any $s \in [3]$ and matrix $U_s \in \mathbb{R}^{q_s \times p_s}$, the marginal multiplication operator $\times_1$ is defined as
$ \mathbf{M}  \times_1 U_1 =  \{\sum_{k = 1}^{p_{1}} \mathbf{M}_{k, j, l} (U_1)_{i, k} \}_{i \in [q_1], \, j \in [p_2], \, l \in [p_3]} \in \mathbb{R}^{q_1 \times p_2 \times p_3}$. 
Marginal multiplications $\times_2$ and $\times_3$ are defined similarly.

The remainder of the paper is organized as follows.  Section~\ref{sec:optimal_localization} introduces the D-MRDPG model and our two-stage change point localization procedure, along with its theoretical guarantees. In Section~\ref{sec:limiting_distributions}, we derive the limiting distributions of the refined estimators and propose a data-driven procedure for constructing confidence intervals. Section~\ref{sec:experiment} presents extensive numerical experiments illustrating the empirical performance of the proposed methods. We conclude with a brief discussion in Section~\ref{sec:conclusion}. All proofs and auxiliary results are deferred to the Appendix.

\section{Change point localization}\label{sec:optimal_localization}

\subsection{Problem formulation}\label{sec:model}

We begin with the multilayer random dot product graph (MRDPG) model \citep{jones2020multilayer}, which generalizes the random dot product graph \citep{young2007random} to multilayer networks. Each layer is characterized by a distinct weight matrix, while all layers share a common set of latent positions.  We focus on undirected edges, as the directed case is similar and thus omitted.

\begin{definition}[Multilayer random dot product graphs, $\mbox{MRDPGs}$]\label{umrdpg-f}
 Given a sequence of deterministic matrices $\{W_{(l)}\}_{l = 1}^L \subset \R^{d \times d}$, let $\{X_i\}_{i=1}^{n} \subset \mathbb{R}^d$ be fixed vectors satisfying $X_i^{\top} W_{(l)} X_j \in [0, 1]$ for all $i, j \in [n], l \in [L]$.
 An adjacency tensor $\mathbf{A}  \in \{0, 1\}^{n \times n  \times L}$ follows an MRDPG if 
\begin{align*}
    \mathbb{P} \{\mathbf{A} \} & = \prod_{l = 1}^{L} \prod_{ 1 \leq i\leq j \leq n}\mathbf{P}_{i, j, l}^{\mathbf{A}_{i, j, l}} (1- \mathbf{P}_{i, j, l})^{1- \mathbf{A} _{i, j, l}} \\
    & = \prod_{l = 1}^{L} \prod_{ 1 \leq i \leq j \leq n} \big(X_i^{\top} W_{(l)} X_j\big)^{\mathbf{A} _{i, j, l}} \big(1 - X_i^{\top} W_{(l)} X_j\big)^{1 - \mathbf{A} _{i, j, l}}.
\end{align*}
We write $\mathbf{A} \sim \mathrm{MRDPG}(\{X_i\}_{i=1}^{n}, \{W_{(l)}\}_{l\in [L]})$ and denote the probability tensor by $\mathbf{P}  \in \mathbb{R}^{n \times n \times L}$. 
\end{definition}

We now extend this static model to a dynamic setting and introduce a change point framework.

\begin{definition}[Dynamic multilayer random dot product graphs, D-MRDPGs]\label{def-umrdpg-f-dynamic}
Let $\{X_i\}_{i=1}^{n} \subset \mathbb{R}^d$ be latent positions and $\{W_{(l)}(t)\}_{l\in[L], t\in [T]} \subset \mathbb{R}^{d\times d}$ be a weight matrix sequence.  A sequence of mutually independent adjacency  tensors  $\{\mathbf{A}(t)\}_{t \in [T]}$ follows the dynamic MRDPGs if 
$ \mathbf{A}(t) \sim \mathrm{MRDPG}( \{X_i\}_{i=1}^{n}, \{W_{(l)}(t)\}_{l \in [L]})$ for $ t \in [T]$. We write 
\[
\{\mathbf{A}(t)\}_{t=1}^T \sim \mathrm{D}\mbox{-}\mathrm{MRDPGs}(\{X_i\}_{i=1}^{n},   \{ \{W_{(l)}(t)\}_{l\in [L]} \}_{t=1}^T),
\]
and write $\{\mathbf{P}(t)\}_{t=1}^T$ as the corresponding  sequence of probability tensors. 
\end{definition}

\begin{model}\label{model-1}
   Let $\{\mathbf{A}(t)\}_{t \in [T]} \subset \{0,1\}^{n \times n \times L}$ follow $\mathrm{D}\mbox{-}\mathrm{MRDPGs}$  as in Definition~\ref{def-umrdpg-f-dynamic}.
\begin{enumerate}[ $(i)$]
    \item  Assume that there exist change points $0 = \eta_0 < \eta_1 < \dots < \eta_K < T = \eta_{K+1}$ such that for $t \in  [T-1]$,    $\{  W_{(l)}(t)\}_{l=1}^L \neq \{  W_{(l)}(t + 1 )\}_{l=1}^L$ if and only if $t \in\{ \eta_k\}_{k=1}^K$.
     Let $\Delta = \min_{k \in [K+1]} (\eta_{k} - \eta_{k-1})$ be the minimal spacing between two consecutive change points and assume $\Delta = \Theta(T)$.
     \item For each $k \in [K]$, define the $k$-th jump size and normalized jump tensor as 
     \[
     \kappa_k  = \| \mathbf{P}(\eta_{k+1}) - \mathbf{P}(\eta_{k}) \|_{\mathrm {F}} \quad \mbox{and}  \quad \boldsymbol{\mathbf{\Psi}}_k = \kappa_k^{-1} \{\mathbf{P}(\eta_{k+1}) - \mathbf{P}(\eta_{k}) \},
     \] 
     and let $\kappa = \min_{k \in [K] } \kappa_k$ denote the smallest jump magnitude. 
\end{enumerate}

\end{model}

Model~\ref{model-1} allows abrupt changes in the layers' connectivity patterns encoded in the weight matrices, while the latent positions remain unchanged over time.  This framework is motivated by a range of practical applications.
For example, in air transportation networks (see Section~\ref{sec:real}), nodes represent airports whose intrinsic attributes, such as geographical location and logistical capacity, are relatively stable.  In contrast, airline routing preferences reflected in the weight matrices may change due to factors such as seasonal demand fluctuations, route optimization strategies or policy interventions.

In Model~\ref{model-1}$(i)$, we assume that the minimal spacing $\Delta$ between successive change points is of the same order as the total time horizon $T$, which essentially bounds the number of change points $K$. This assumption can be relaxed, as discussed further in Section~\ref{sec:conclusion}. In Model~\ref{model-1}$(ii)$, the magnitude of a change is quantified via the Frobenius norm of the difference between expected adjacency tensors. This metric is sufficiently general to accommodate both dense changes - small but widespread deviations across many layers - and sparse changes - large deviations concentrated in a few layers. Throughout, we allow all model parameters, including the number of nodes $n$, number of layers $L$, latent dimension $d$, jump size $\kappa$ and minimal spacing $\Delta$ to diverge with $T$.

\subsection{Change point localization algorithm}\label{sec:algorithm}

In this section, we introduce a two-stage procedure for offline change point localization in dynamic multilayer networks, detailed in Algorithm~\ref{offline-algorithm}. 
\textbf{Stage I} generates a coarse set of change point candidates using seeded binary segmentation and CUSUM statistics. \textbf{Stage II} refines them via localized scan statistics constructed using a tensor-based low-rank estimation technique. This approach builds on \cite{wang2021optimal} for single-layer networks and extends it to the multilayer setting.

For \textbf{Stage I}, we begin by defining the seeded intervals \citep{kovacs2023seeded}  and CUSUM statistics \citep{page1954continuous} for dynamic multilayer networks in Definitions~\ref{def-seeded} and \ref{def-cusum-f}.

\begin{definition}[Seeded intervals]\label{def-seeded}
Let $J = \lceil C_{J} \log_2(T) \rceil$ for some sufficiently large absolute constant $C_{J} > 0$.
For each $j \in [J]$, define the collection of intervals  $\mathcal{J}_j$ as
\[
\mathcal{J}_j = \{ ( \lfloor (i - 1) T 2^{-j} \rfloor, \lceil (i - 1) T 2^{-j} + T 2^{-j+1} \rceil ] \colon  i \in  [2^j - 1] \}.
\]
The full collection of seeded intervals is defined as $\mathcal{J} = \bigcup_{j=1}^{J} \mathcal{J}_j$.
\end{definition}

\begin{definition}[CUSUM statistics]\label{def-cusum-f}
Given a tensor sequence $\{\mathbf{B}(t)\}_{t \in [T]}$ and any $0 \leq s < t < e \leq T$, define the CUSUM statistics as
\begin{equation}\label{def-omega}
    \widetilde{\mathbf{B}}^{s,e}(t) = \sum_{u=s+1}^{e}  \omega_{s, e}^t(u)  \mathbf{B}(u),  \quad \mbox{where} \quad \omega_{s, e}^t(u) = \begin{cases}
 	  \sqrt{\frac{e-t}{(e-s)(t-s)}},   \quad & \mbox{for }u 
      \in [t]\backslash[s],\\
  - \sqrt{ \frac{t-s}{(e-s)(e-t)}},   \quad & \mbox{for }u
  \in [e] \backslash [t].
     \end{cases}  
\end{equation}

\end{definition}

\textbf{Stage I} implements a modified version of seeded binary segmentation (SBS), a computationally efficient algorithm introduced by \cite{kovacs2023seeded}. SBS leverages seeded intervals to construct a multiscale collection of candidate regions for detecting multiple change points. Within each interval, the algorithm computes CUSUM statistics and retains time points where the statistic is maximized and exceeds a predefined threshold, as preliminary change point estimators

We next define the refined scan statistics used in \textbf{Stage II}, based on tensor heteroskedastic principal component analysis (TH-PCA), a low-rank tensor estimation method proposed by \cite{han2022optimal} and detailed in Algorithm~\ref{thpca} in Appendix~\ref{section:add-algorithm}.

\begin{definition}[Refined scan statistics]\label{def-cusum-new} 
Let $\{\mathbf{A}'(t)\}_{t \in [T]}$ and $\{\mathbf{B}'(t)\}_{t \in [T]}$ be independent sequences generated according to Definition~\ref{def-umrdpg-f-dynamic}. Given $\{(b_k, s_k, e_k)\}_{k=1}^{\widetilde{K}}$,   for any  $k \in [\widetilde{K}]$ and $t \in (s_k, e_k)$, we define the refined scan statistic as
\[
    \widehat{D}^{s_k, e_k}_{b_k}(t) =  \big \vert 
 \big\langle \widehat{\mathbf{P}}^{s_k,e_k}(b_k) /  \|\widehat{\mathbf{P}}^{s_k,e_k}(b_k) \|_{\mathrm{F}}, \widetilde{\mathbf{A}'}^{s_k,e_k} (t) \big\rangle  \big \vert, 
\]
where 
\[
\widehat{\mathbf{P}}^{s_k, e_k}(b_k) = \mathrm{TH}\text{-}\mathrm{PCA}\bigg(\widetilde{\mathbf{B'}}^{s_k,e_k}(b_k), (d, d, m^{s_k, e_k}_{b_k}),  \sqrt{\frac{(e_k-b_k) (b_k-s_k)}{e_k-s_k}},    \sqrt{\frac{(e_k-b_k) (b_k-s_k)}{e_k-s_k}} \bigg)
\]
with $\mathrm{TH}\mbox{-}\mathrm{PCA}$ detailed in Algorithm~\ref{thpca}, $\widetilde{\mathbf{B'}}^{\cdot, \cdot} (\cdot)$ defined in Definition \ref{def-cusum-f} and $m^{s, e}_{b_k}$  defined in Assumption~\ref{ass_X_f_Q_f}$(ii)$.
\end{definition}

\textbf{Stage II} refines each preliminary change point estimate from \textbf{Stage I} by locating the time point that maximizes a refined scan statistic within a local window around the initial estimate.  This step leverages the TH-PCA procedure with an additional truncation step (see Algorithm~\ref{thpca}) to more accurately estimate the local expected CUSUM adjacency tensors, leading to provably improved localization accuracy.

The assumption of mutual independence among all four sequences in Algorithm~\ref{offline-algorithm} is imposed for theoretical convenience. In practice (and in our numerical experiments in Section~\ref{sec:experiment}), \textbf{Stage I} and \textbf{Stage II} are implemented using the same two split tensor sequences.

\begin{algorithm}[t] 
\caption{Two-stage change point localization for D-MRDPGs} \label{offline-algorithm}
\begin{algorithmic}
\INPUT{Mutually independent sequences $\{\mathbf{A}(t)\}_{t \in [T]}, $ $\{\mathbf{A}'(t)\}_{t \in [T]}, \{\mathbf{B}(t)\}_{t \in [T]},  \{\mathbf{B}'(t)\}_{t \in [T]}  \subset \{ 0, 1\}^{n \times n \times L}$, threshold $\tau \in \mathbb{R}^+$, collection of seeded intervals $\mathcal{J}$ }
\Initialise{$s \leftarrow 0$, $e \leftarrow T$, $\widetilde{\mathcal{C}} \leftarrow \emptyset$}
\vspace{1mm}
\Statex{\textbf{Stage I:} Seeded Binary Segmentation, $\mathrm{SBS}\big((s, e), \tau, \mathcal{J} \big)$}
   \For{$\mathcal{I} = (\alpha', \beta'] \in \mathcal{J}$}
\If{$\mathcal{I} = (\alpha', \beta'] \subseteq (s, e]$
}
\State{$(\alpha, \beta] = ( \lfloor \alpha'  + 64^{-1}(\beta' - \alpha') \rfloor, \lceil \beta' - 64^{-1}(\beta' - \alpha') \rceil ]$
}
\If{$\beta - \alpha \geq 2$}
\State{$b_{\mathcal{I}} \leftarrow \argmax_{ \alpha < t  < \beta }  
     \big \vert \big\langle \widetilde{\mathbf{A}}^{\alpha,\beta}(t) , 
    \widetilde{\mathbf{B}}^{\alpha,\beta} (t) \rangle \big\vert, \, 
    a_{\mathcal{I}} \leftarrow    \big \vert \big\langle \widetilde{\mathbf{A}}^{\alpha,\beta}(b_{\mathcal{I}}) , 
    \widetilde{\mathbf{B}}^{\alpha,\beta} (b_{\mathcal{I}}) \rangle \big\vert$}
\Else{ $a_{\mathcal{I}}  \leftarrow -1$ }
\EndIf
\Else{ $a_{\mathcal{I}}  \leftarrow -1$ }
\EndIf
\EndFor
\State{$\mathcal{I}^* \leftarrow \argmax_{\mathcal{I} \in \mathcal{J}} a_{\mathcal{I}}$}
\If{$a_{\mathcal{I}^*} > \tau $}
\State{$\widetilde{\mathcal{C}} \leftarrow\widetilde{\mathcal{C}} \cup \{ b_{\mathcal{I}^*}\} $, $\mathrm{SBS}\big((s, b_{\mathcal{I}^*}), \tau, \mathcal{J} \big)$, $\mathrm{SBS} \big((b_{\mathcal{I}^*}, e), \tau, \mathcal{J} \big)$}
\EndIf
\vspace{1.5mm}
\Statex{\textbf{Stage II:} Local Refinement, $\mathrm{LR} ( \hspace{0.3mm} \widetilde{\mathcal{C}} \hspace{0.3mm} )$} 
\State{$\{b_k\}_{k=1}^{\widetilde{K}} \leftarrow \widetilde{\mathcal{C}}$ with $0 = b_0 < b_1 < \cdots  <b_{\widetilde{K}} < b_{\widetilde{K}+1} = T$}
\For{$k = 1$ to $\widetilde{K}$}
\State{$(s_k, e_k] \leftarrow \big( \lfloor (b_{k-1} + b_k)/2 \rfloor, \lceil (b_k + b_{k+1})/2 \rceil \big]$}
\State{$\widetilde{\eta}_k \leftarrow \argmax_{s_k < t < e_k} \widehat{D}_{b_k}^{s_k,e_k}(t)$  \Comment{See Definition~\ref{def-cusum-new}}}
    \EndFor
\vspace{2mm}
\OUTPUT{$\{\widetilde{\eta}_k\}_{k=1}^{\widetilde{K}}$}
\end{algorithmic}
\end{algorithm}

\subsection{Theoretical guarantees}\label{sec:theory}

This section establishes the theoretical guarantees of the proposed two-stage change point localization procedure (Algorithm~\ref{offline-algorithm}).  
We begin by justifying the use of low-rank tensor estimation via TH-PCA (Algorithm \ref{thpca}) in \textbf{Stage II} through an analysis of the expected CUSUM-transformed and average adjacency tensors. While the expected averaged adjacency tensors introduced below are not used in this section, they are essential for deriving the limiting distributions in Section~\ref{sec:limiting_distributions}. 

For any $ 0 \leq s < t< e \leq T$, define the expected CUSUM-transformed and average adjacency tensors as
\begin{equation}\label{def-average}
\widetilde{\mathbf{P}}^{s, e }(t) = \E \big\{ \widetilde{\mathbf{B}}^{s, e }(t) \big\} \quad \mbox{and} \quad 
{\mathbf{P}}^{s, e } = \E \big\{  \mathbf{B}^{s, e} \big\}, \quad \mbox{where} \quad   \mathbf{B}^{s, e} = (e-s)^{-1} \sum_{t=s+1}^e \mathbf{B}(t),
\end{equation}
and $\widetilde{\mathbf{B}}^{\cdot, \cdot} (\cdot)$ is defined in Definition \ref{def-cusum-f}.
Both tensors admit Tucker representations of the form
$
\widetilde{\mathbf{P}}^{s, e }(t) = \mathbf{S} \times_1 X \times_2 X \times_3 \widetilde{Q}^{s, e}(t),$ and ${\mathbf{P}}^{s, e } = \mathbf{S} \times_1 X \times_2 X \times_3 Q^{s, e}$,
where $X = (X_1, \ldots, X_{n})^{\top} \in \mathbb{R}^{n \times d}$  and 
$\mathbf{S} \in \R^{d \times d \times d^2}$ with $\mathbf{S}_{i, j, l} = \mathbbm{1}\{l = (i-1)d + j\}$.
The matrices $\widetilde{Q}^{s,e}(t)$ and $Q^{s, e}$ are given by 
\begin{equation}\label{def-tilde-Q}
\widetilde{Q}^{s,e}(t) = \sum_{u=s+1}^{e}   \omega_{s, e}^t(u) Q (u), \quad  Q^{s, e}  = (e-s)^{-1}\sum_{t=s+1}^e   Q(t),
\end{equation}
where $\omega_{s, e}^t(u)$ is define in \eqref{def-omega} and $ Q(u) \in \R^{L \times d^2} $ with rows
\begin{equation}\label{def-Q}
\big( Q(u) \big)_l = \big( (W_{(l)}(u))_1  \cdots  (W_{(l)}(u))_d \big), \quad l \in [L].
\end{equation}
To establish the low-rank structure of $\widetilde{\mathbf{P}}^{s, e}(t)$ and $\mathbf{P}^{s, e}$ (in terms of Tucker ranks, see Section~\ref{sec:notation}), and to state theoretical guarantees for Algorithm~\ref{offline-algorithm}, we state some necessary assumptions below.  

\begin{assumption}\label{ass_X_f_Q_f}
Consider $ \mathrm{D}\mbox{-}\mathrm{MRDPGs}(\{X_i\}_{i=1}^{n}, \{ \{W_{(l)}(t)\}_{l\in [L]} \}_{t=1}^T)$ from  Definition~\ref{def-umrdpg-f-dynamic}. 
\begin{enumerate}[$(i)$]
    \item  Let $X = (X_1, \ldots, X_{n})^{\top} \in \mathbb{R}^{n \times d}$. Assume that  $
\mathrm{rank}(X) = d$, $\sigma_1(X)/\sigma_d(X) \leq C_{\sigma}$ and $\sigma_{d}(X) \geq C_{\mathrm{gap}} \sqrt{n}$ with  absolute constants $C_{\sigma}, C_{\mathrm{gap}} >0$.
\item For any $0 \leq s < t < e \leq T$, let $\widetilde{Q}^{s, e}(t) \in \mathbb{R}^{L \times d^2}$ be defined in \eqref{def-tilde-Q}. Denote $m^{s, e}_t = \mathrm{rank}(\widetilde{Q}^{s, e}(t))$.  Assume that
$\sigma_1\big(\widetilde{Q}^{s, e}(t) \big)/\sigma_{m^{s, e}_t} \big( \widetilde{Q}^{s, e}(t) \big) \leq C_{\sigma}$ and $\sigma_{m^{s, e}_t}\big( \widetilde{Q}^{s, e}(t) \big) \geq C_{\mathrm{gap}}$ with  absolute constants $C_{\mathrm{gap}}, C_{\sigma}>0$.
\item  For any $0 \leq s <  e \leq T$, let $Q^{s, e} \in \mathbb{R}^{L \times d^2}$ be defined in \eqref{def-tilde-Q}. Denote $m^{s, e}= \mathrm{rank}(Q^{s, e})$. Assume that 
$\sigma_1\big(Q^{s, e} \big)/\sigma_{m^{s, e}} \big( Q^{s, e} \big) \leq C_{\sigma}$ and $\sigma_{m^{s, e}}\big( Q^{s, e} \big) \geq C_{\mathrm{gap}}$
with  absolute constants $C_{\mathrm{gap}}, C_{\sigma}>0$. 
\end{enumerate}

\end{assumption}

Assumption~\ref{ass_X_f_Q_f}$(i)$ imposes a full-rank condition on the latent position matrix $X$, requiring its smallest singular value to be at least of order $\sqrt{n}$, with all singular values of the same order. Since $X$ represents latent positions rather than observed data, the full-rankness of $X$ can be interpreted as a condition on the knowledge of the intrinsic dimension $d$, ensuring that the input dimension to TH-PCA is no smaller than the true latent dimension $d$. Further discussion on rank selection, see \cite{wang2023multilayer}.

Assumptions~\ref{ass_X_f_Q_f}$(ii)$ and $(iii)$ -  with $(iii)$ for Section~\ref{sec:limiting_distributions} -  impose low-rank conditions on the CUSUM and averaged forms of $\{Q(t)\}_{t=1}^T$,  where each $Q(t)$ comprises the weight matrices $\{W_{(l)}(t)\}_{l=1}^L$.
In Appendix~\ref{proof-theorem-2}, we show that, with high probability,  each working interval $(s_k,e_k]$ or $(\tilde{s}_k,\tilde{e}_{k}]$ contains exactly one change point $\eta_k$, implying  $ \max\{m^{s_k, e_k}_t, m^{\tilde{s}_k,\tilde{e}_k}\} \leq \mathrm{rank}(Q(\eta_k)) + \mathrm{rank}(Q(\eta_{k+1}))$ for $t \in (s_k, e_k)$.
This implicitly constraints the ranks of  $\{Q(\eta_k)\}_{k=1}^{K+1}$.
While this low-rank structure may not directly or transparently reflect the explicit model structure, such ambiguity is common in tensor-based models \citep[e.g.][]{jing2021community}.

The signal-to-noise ratio (SNR) is commonly used to characterize the inherent difficulty of change point detection. We now state the SNR condition required for our theoretical guarantees. 

\begin{assumption}[Signal-to-noise ratio condition]\label{ass-SNR}
Assume that there exists a large enough absolute constant $C_{\mathrm {SNR}}>0$ such that
\[
 \kappa \sqrt{\Delta}  \geq C_{\mathrm{SNR}} \log(T) \sqrt{n L^{1/2} + d^2m_{\max} +  nd+Lm_{\max} },
\]
where $m_{\max} = \max_{k \in [K+1]} \mathrm{rank}\big(Q(\eta_k)\big)$ with $Q(\eta_k)$ defined in \eqref{def-Q}.
\end{assumption} 

We compare Assumption~\ref{ass-SNR} to its counterpart in \cite{wang2021optimal}. When the sparsity parameter $\rho = 1$, their SNR condition (Assumption 3) becomes $\kappa \sqrt{\Delta} \geq C_{\mathrm{SNR}} \log^{1+\xi}(T) \sqrt{nd} $ for some $\xi > 0$. Our assumption is consistent with this and extends it to the multilayer setting by accounting for the additional complexity from multilayers and the low-rank structure of layers' weight matrices.

\begin{theorem}\label{theorem-2}
Let $\{ \widetilde{\eta}_k \}_{k=1}^{\widetilde{K}}$ be the output of Algorithm~\ref{offline-algorithm}.
Suppose the mutually independent adjacency tensor sequences $\{\mathbf{A}(t)\}_{t \in [T]}, \{\mathbf{A}'(t)\}_{t \in [T]},  \{\mathbf{B}(t)\}_{t \in [T]}, \{\mathbf{B}'(t)\}_{t \in [T]} \subset \{0, 1\}^{n \times n \times L}$ are generated according to Definition~\ref{def-umrdpg-f-dynamic} and satisfy Model~\ref{model-1}, Assumptions~\ref{ass_X_f_Q_f}$(i)$, $(ii)$ and \ref{ass-SNR}. 
Assume the threshold $\tau$ is chosen such that $c_{\tau,1}   n\sqrt{L} \log^{3/2}(T) <   \tau  < c_{\tau,2} \kappa^2 \Delta$,
where $c_{\tau,1}, c_{\tau,2} > 0$ are sufficiently large and small absolute constants, respectively. We have that 
\[
\mathbb{P} \Big\{ \widetilde{K} = K \mbox{ and } \vert \widetilde{\eta}_k - \eta_k \vert  \leq \epsilon_k, \, \forall k \in [K] \Big\} \geq 1 - CT^{-c}, \quad \mbox{where }  \epsilon_k = C_{\epsilon}      \frac{\log (T)}{\kappa_k^2}, 
\]
and $C_{\epsilon}, C, c >0$ are absolute constants. 
\end{theorem}

Theorem~\ref{theorem-2}  implies that, with probability tending to $1$ as $T \to \infty$, 
the estimated number of change points satisfies $\widetilde{K} = K$ and the relative localization error vanishes:
\[
\max_{k \in [K]} \Delta^{-1} |\widetilde{\eta}_k - \eta_k| 
\leq C_\epsilon  \Delta^{-1} \kappa^{-2} \log(T)  \to 0
\]
by Assumption~\ref{ass-SNR}. This establishes the consistency of Algorithm~\ref{offline-algorithm}  in both detecting and localizing all change points.
The localization error rate matches the minimax-optimal rate up to a logarithmic factor established for single-layer networks in \cite{wang2021optimal}, showing that our method generalizes these guarantees to the more complex multilayer regime without sacrificing localization accuracy. Moreover, compared to the online framework for D-MRDPGs studied in \cite{wang2023multilayer}, which yields a localization error rate of order $ \kappa^{-2} (d^2m_{\max} + nd + Lm_{\max}) \log(\Delta/\alpha)$, where $\alpha$ controls the Type-I error rate, our two-stage procedure achieves a significantly sharper rate in the offline setting.

\section{Limiting distributions}\label{sec:limiting_distributions}

Inference on change points is generally more challenging than establishing high-probability bounds on localization errors. To address this, we introduce a final refinement step, inspired by approaches such as those in \cite{madrid2023change, xue2024change, xu2024change}.

Let $\{\mathbf{A}(t)\}_{t \in [T]}$ and $ \{\mathbf{B}(t)\}_{t \in [T]} $ be independent samples as defined in Definition~\ref{def-umrdpg-f-dynamic}. Let $\{\widetilde{\eta}_k\}_{k=1}^{K}$ be the output of Algorithm~\ref{offline-algorithm} with $0 =\widetilde{\eta}_0 < \widetilde{\eta}_1 < \cdots < \widetilde{\eta}_{\widetilde{K}} <  \widetilde{\eta}_{\widetilde{K}+1} = T $. For each $k \in [\widetilde{K}]$, define the final estimators as
\begin{equation}\label{def-estimator-final}
    \widehat{\eta}_k = \argmin_{\tilde{s}_k < t < \tilde{e}_k} \mathcal{Q}_k(t)  =   \argmin_{\tilde{s}_k < t < \tilde{e}_k}\sum_{u=\tilde{s}_k+1}^{t} \| \mathbf{A}(u) - \widehat{\mathbf{P}}^{\widetilde{\eta}_{k-1}, \widetilde{\eta}_{k}} \|_{\mathrm{F}}^2 + \sum_{u=t +1}^{\tilde{e}_k} \| \mathbf{A}(u) - \widehat{\mathbf{P}}^{\widetilde{\eta}_{k}, \widetilde{\eta}_{k+1}} \|_{\mathrm{F}}^2 ,
\end{equation}
where $(\tilde{s}_k,  \tilde{e}_k ] = ((\widetilde{\eta}_{k-1} + \widetilde{\eta}_k)/2,  (\widetilde{\eta}_k + \widetilde{\eta}_{k+1})/2]$ and
\begin{equation}\label{def-estimator-final-P}
\widehat{\mathbf{P}}^{\widetilde{\eta}_{k-1}, \widetilde{\eta}_{k}} = \mathrm{TH}\text{-}\mathrm{PCA} ( \mathbf{B}^{\widetilde{\eta}_{k-1}, \widetilde{\eta}_{k}},  (d, d,  m^{\widetilde{\eta}_{k-1}, \widetilde{\eta}_{k}}), 1, 0 ),
\end{equation}
with $\mathrm{TH}\mbox{-}\mathrm{PCA}$ detailed in Algorithm~\ref{thpca},  $\mathbf{B}^{\cdot , \cdot}$ defined in \eqref{def-average} and $m_{\widetilde{\eta}^{k-1}, \widetilde{\eta}_{k}}$ defined in Assumption~\ref{ass_X_f_Q_f}$(iii)$.

\begin{theorem}\label{theorem-inference}

Let  $\{\mathbf{A}(t)\}_{t \in [T]},\{\mathbf{A}'(t)\}_{t \in [T]},  \{\mathbf{B}(t)\}_{t \in [T]},  \{\mathbf{B}'(t)\}_{t \in [T]} \subset \{0, 1\}^{n \times n \times L} $ be mutually independent adjacency tensor sequences generated according to Definition~\ref{def-umrdpg-f-dynamic} and satisfying Model~\ref{model-1}, Assumptions~\ref{ass_X_f_Q_f} and \ref{ass-SNR}.  Let $\{ \widehat{\eta}_k \}_{k=1}^{\widetilde{K}}$ be defined in \eqref{def-estimator-final} with  $\{ \widetilde{\eta}_k \}_{k=1}^{\widetilde{K}}$ obtained from Algorithm~\ref{offline-algorithm}, using a threshold $\tau$ satisfying condition stated in Theorem~\ref{theorem-2}.

For $ k \in [K] $, if $ \kappa_k \to 0 $, as $ T \to \infty $, then when $ T \to \infty $, we have $ \vert \widehat{\eta}_k - \eta_k\vert  = O_p(\kappa_k^{-2})$ and
    \[
    \kappa_k^2 (\widehat{\eta}_k - \eta_k) \xrightarrow{\mathcal{D}} \argmin_{r \in \mathbb{R}} \mathcal{P}'_k(r), \quad  \mbox{where} \quad   \mathcal{P}'_k(r) =
    \begin{cases}
        - r  + 2 \sigma_{k, k} \mathbb{B}_1(-r), & r < 0, \\
        0, & r = 0, \\
        r  + 2 \sigma_{k, k+1} \mathbb{B}_2(r), & r > 0,
    \end{cases}
    \]
    for $ r \in \mathbb{Z} $. Here, $\mathbb{B}_1(r) $ and $\mathbb{B}_2(r) $ are independent standard Brownian motions, and for any $k' \in \{k , k+1\}$,  $\sigma_{k, k'}^2 =  \operatorname{Var} \big( \langle \boldsymbol{\mathbf{\Psi}}_k , \mathbf{E}_{k'}(1) \rangle \big)$, where $\boldsymbol{\mathbf{\Psi}}_k$ is the normalized jump tensor (Model~\ref{model-1}), and 
     $\mathbf{E}_{k'}(t) = \mathbf{A}_{k'}(t) - \mathbf{P}(\eta_{k'})$ with $\{\mathbf{A}_{k'}(t)\}_{t \in \mathbb{Z}} \overset{\mathrm{i.i.d.}}{\sim}  \mathrm{MRDPG}(\{X_i\}_{i=1}^{n}, \{W_{(l)}(\eta_{k'})\}_{l\in [L]})$.
\end{theorem} 

Theorem~\ref{theorem-inference} establishes the localization error bounds and limiting distributions for the refined change point estimators in the vanishing jump regime ($ \kappa_k \to 0 $). 
In particular, it shows the uniform tightness
$
\kappa_k^2 |\widehat{\eta}_k - \eta_k| = O_p(1),
$
which improves upon Theorem~\ref{theorem-2} by a logarithmic factor and guarantees the existence of limiting distributions.
To the best of our knowledge, Theorem~\ref{theorem-inference} is the first to derive limiting distributions for change point estimators in network data.  These limiting distributions are associated with a two-sided Brownian motion. Results for the non-vanishing jump regime ($\kappa_k \to \rho_k > 0$) are deferred to Appendix~\ref{sec:add-theory}.

\subsection{Confidence interval construction}\label{sec:CI}

Using Theorem~\ref{theorem-inference}, we construct $(1-\alpha)$ confidence intervals for $\eta_k$, $k \in [K]$, in the vanishing regime, for a user-specified confidence level $\alpha \in (0, 1)$.

\noindent \textbf{Step 1: Estimate the jump size and normalized jump tensor.}
Compute the estimated jump size 
\[
\hat{\kappa}_k = \| \widehat{\mathbf{P}}^{\widetilde{\eta}_{k}, \widetilde{\eta}_{k+1}} - \widehat{\mathbf{P}}^{\widetilde{\eta}_{k-1}, \widetilde{\eta}_{k}} \|_{\mathrm{F}},
\]
and the estimated normalized jump tensor 
\[
\widehat{\boldsymbol{\mathbf{\Psi}}}_k = \hat{\kappa}_k^{-1}  ( \widehat{\mathbf{P}}^{\widetilde{\eta}_{k}, \widetilde{\eta}_{k+1}} - \widehat{\mathbf{P}}^{\widetilde{\eta}_{k-1}, \widetilde{\eta}_{k}} ),\]
where $\widehat{\mathbf{P}}^{\cdot, \cdot}$ is defined in \eqref{def-estimator-final-P}.

\noindent \textbf{Step 2: Estimate the variances.}
For each $k' \in \{k, k+1\}$, compute
\[
\hat{\sigma}_{k, k'}^2 = \frac{1}{\widetilde{\eta}_{k'} - \widetilde{\eta}_{k'-1} - 1} \sum_{t=\widetilde{\eta}_{k'-1}+1}^{\widetilde{\eta}_{k'}} \big( \langle \boldsymbol{\widehat{\mathbf{\Psi}}}_{k} , \mathbf{A}(t) - \widehat{\mathbf{P}}^{\widetilde{\eta}_{k'-1}, \widetilde{\eta}_{k'}} \rangle \big)^2.
\]
\textbf{Step 3: Simulate limiting distributions.}
Let $B \in  \mathbb{N}^+$ and $M \in \R^+$. For each $b \in [B] $, let
\[
\hat{u}^{(b)}_k = \argmin_{r \in (-M, M)}   \widehat{\mathcal{P}}'_k(r),
\quad \mbox{where} \quad 
   \widehat{\mathcal{P}}'_k(r) =
    \begin{cases}
        - r  +   \frac{2 \hat{\sigma}_{k, k}}{\sqrt{T}} \sum_{i = \lceil Tr \rceil}^{-1} z_i^{(b)} , & r < 0, \\
        0, & r = 0, \\
        r   +  \frac{ 2 \hat{\sigma}_{k, k+1}}{\sqrt{T} }\sum_{i = 1}^{\lfloor Tr \rfloor} z_i^{(b)} , & r > 0,
    \end{cases}
\]
with independent standard Gaussian random variables $\{ z_i^{(b)} \}_{i = -\lfloor TM \rfloor}^{\lceil TM \rceil}$. 

\noindent \textbf{Step 4: Construct the confidence interval.} Let  $\hat{q}_{\alpha/2}, \hat{q}_{1 - \alpha/2}$ be empirical quantiles of $\{\hat{u}_k^{(b)}\}_{b=1}^B$. The $(1-\alpha)$ confidence interval for $\eta_k$ is given by
\[
\bigg[ \widehat{\eta}_k - \frac{\hat{q}_{1 - \alpha/2}}{\hat{\kappa}_k^2} \mathbbm{1}\{\hat{\kappa}_k \neq0 \}, \, \widehat{\eta}_k -  \frac{\hat{q}_{\alpha/2}}{\hat{\kappa}_k^2} \mathbbm{1}\{\hat{\kappa}_k \neq0 \} \bigg].
\]

The empirical performance of this procedure is evaluated in Section~\ref{sec:simulation}.

\section{Numerical experiments}\label{sec:experiment}

We evaluate the proposed method on synthetic and real data in Sections~\ref{sec:simulation} and \ref{sec:real}, respectively.

\subsection{Simulation studies}\label{sec:simulation}

\begin{table}
\caption{Means of evaluation metrics for networks simulated from Scenarios 1 and 2. }
\vspace{1em}
\label{tbl:s1s2}
\centering
\resizebox{\textwidth}{!}{ 
\begin{tabular}{ll llll  llll}
\toprule 
&  & \multicolumn{4}{c}{Scenario 1} & \multicolumn{4}{c}{Scenario 2} \\
$n$ & Method & $|\widehat{K}-K|\downarrow$ & $d(\widehat{\mathcal{C}}|\mathcal{C})\downarrow$ & $d(\mathcal{C}|\widehat{\mathcal{C}})\downarrow$ & $C(\mathcal{G},\mathcal{G'})\uparrow$ & $|\widehat{K}-K|\downarrow$ & $d(\widehat{\mathcal{C}}|\mathcal{C})\downarrow$ & $d(\mathcal{C}|\widehat{\mathcal{C}})\downarrow$ & $C(\mathcal{G},\mathcal{G'})\uparrow$\\ 
\midrule
\multirow{5}{2em}{$50$} 
& CPDmrdpg       & $0.01$ & $0.00$ & $0.42$ & $99.86\%$                & $0.00$ & $0.00$ & $0.00$ & $100\%$ \\
& gSeg (nets.)   & $1.09$ & $\text{Inf}$ & $\text{Inf}$ & $52.82\%$    & $1.60$ & $\text{Inf}$ & $\text{Inf}$ & $67.68\%$ \\
& kerSeg (nets.) & $0.10$ & $0.00$ & $3.12$ & $99.13\%$                & $0.15$ & $0.00$ & $1.53$ & $99.32\%$ \\
& gSeg (frob.)   & $0.52$ & $\text{Inf}$ & $\text{Inf}$ & $90.12\%$    & $0.23$ & $\text{Inf}$ & $\text{Inf}$ & $97.71\%$ \\
& kerSeg (frob.) & $0.26$ & $0.00$ & $5.76$ & $98.35\%$                & $0.35$ & $0.11$ & $3.43$ & $98.37\%$ \\
\midrule
\multirow{5}{2em}{$100$} 
& CPDmrdpg       & $0.00$ & $0.00$ & $0.00$ & $100\%$                  & $0.00$ & $0.00$ & $0.00$ & $100\%$ \\
& gSeg (nets.)   & $1.12$ & $\text{Inf}$ & $\text{Inf}$ & $52.62\%$    & $1.58$ & $\text{Inf}$ & $\text{Inf}$ & $69.24\%$ \\
& kerSeg (nets.) & $0.12$ & $0.00$ & $2.82$ & $99.17\%$                & $0.16$ & $0.00$ & $1.81$ & $99.31\%$ \\
& gSeg (frob.)   & $0.47$ & $\text{Inf}$ & $\text{Inf}$ & $88.71\%$    & $0.16$ & $0.04$ & $1.65$ & $99.17\%$ \\
& kerSeg (frob.) & $0.30$ & $0.00$ & $6.07$ & $98.11\%$                & $0.40$ & $0.02$ & $4.42$ & $97.81\%$ \\
\bottomrule
\end{tabular}
} 
\end{table}

To evaluate the performance of our method (Algorithm~\ref{offline-algorithm}) for change point detection and localization, we compare it to gSeg~\citep{chen2015gseg} and kerSeg~\citep{song2024practical}. For the competitors, we consider two input types: networks (nets.) and their layer-wise Frobenius norms (frob.). For gSeg, we construct the similarity graph using the minimum spanning tree and apply the original edge-count scan statistics. For kerSeg, we use the kernel-based scan statistics fGKCP$_1$. For both methods, we set the significance level $\alpha=0.05$. Our proposed method is referred to as CPDmrdpg. 
Following \cite{wang2023multilayer}, we use relatively large Tucker ranks as inputs to TH-PCA (Algorithm~\ref{thpca}) for robustness, setting $r_1= r_2 = 15$ and $r_3=L$ to compute the refined scan statistics (Definition~\ref{def-cusum-new}).
Based on Theorem~\ref{theorem-2}, we set the threshold $\tau = c_{\tau,1}   n\sqrt{L} \log^{3/2}(T)  $ with $c_{\tau,1} = 0.1$; justification and sensitivity analysis are provided in Appendix~\ref{appendix_syn_data}.   We also assess the confidence intervals constructed utilizing the procedure in Section~\ref{sec:CI},  a capability not supported by the competitors. We set $B = 500$ and $M = T$ as suggested by \cite{xu2024change}.

Performance is quantified using the following metrics:  
\begin{enumerate} [$(i)$]
    \item Absolute error: $|\widehat{K}-K|$ where $\widehat{K}$ and $K$ denote the numbers of estimated and true change points, respectively;
    \item One-sided Hausdorff distances (see Section~\ref{sec:notation}):
$d(\widehat{\mathcal{C}}|\mathcal{C})$ and $d(\mathcal{C}|\widehat{\mathcal{C}})$
where $\widehat{\mathcal{C}}$ and $\mathcal{C}$ denote the sets of estimated and true change points, respectively; 
\item Time segment coverage: 
$C(\mathcal{G},\mathcal{G'}) = T^{-1} \sum_{\mathcal{A} \in \mathcal{G}} |\mathcal{A}| \cdot \max_{\mathcal{A'} \in \mathcal{G'}} |\mathcal{A} \cap \mathcal{A'}|/ |\mathcal{A} \cup \mathcal{A'}|$
where $\mathcal{G}$ and $\mathcal{G}'$ denote the partitions of the time span $[1, T]$ into intervals between consecutive true and estimated change points, respectively.
\end{enumerate} 

Throughout, we set the time horizon to $T = 200$ and the number of layers to $L = 4$, and consider node sizes $n \in \{50, 100\}$. Each setting is evaluated over 100 Monte Carlo trials. We consider two network models: the Dirichlet distribution model~(DDM) and the multilayer stochastic block model (MSBM), with structural changes specified in each scenario. In the DDM, we generate latent positions 
\[
\{X_i\}_{i=1}^n \cup \{Y_i\}_{i=1}^n \overset{\mathrm{i.i.d.}}{\sim}  \text{Dirichlet}(\mathbf{1}_d)
\]
with $d = 5$ and $\mathbf{1}_d \in \mathbb{R}^d$ denoting the all-one vector. For each time $t$, we sample weight matrices $\{W_{(l)}(t)\}_{l=1}^L \subset \mathbb{R}^{d \times d}$ with entries 
\[
(W_{(l)}(t))_{u,v} \sim \text{Uniform}((\rho_t L + l)/(4L), (\rho_t L + l + 1)/(4L)).
\]
The edge probabilities are given by $\mathbf{P}_{i,j,l}(t) = X_i^\top W_{(l)}(t) Y_j$ and the adjacency entries are sampled as $\mathbf{A}_{i,j,l}(t) \sim \text{Bernoulli}(\mathbf{P}_{i,j,l}(t))$. In the MSBM, the edge probability tensor $\mathbf{P}_{i,j,l}(t) \in [0,1]^{n \times n \times L}$ is defined as $\mathbf{P}_{i,j,l}(t) = p_{1,l}$ if nodes $i, j \in \mathcal{B}_c$ for some $c \in [C_t]$, and $p_{2,l}$ otherwise, where $\{\mathcal{B}_{c}\}_{c \in [C_t]}$ partitions the nodes into $C_t$ communities. The connection probabilities are drawn from 
\[
p_{1,l} \sim \text{Uniform}((3L+l-1)/(4L), (3L+l)/(4L))
\]
and 
\[
p_{2,l} \sim \text{Uniform}((2L+l-1)/(4L), (2L+l)/(4L)).
\] 
The adjacency tensor $\mathbf{A}(t) \in \{0,1\}^{n \times n \times L}$ is then sampled $\mathbf{A}_{i,j,l}(t) \overset{\mathrm{ind.}}{\sim} \text{Bernoulli}(\mathbf{P}_{i,j,l}(t))$.

\medskip
\noindent\textbf{Scenario 1.} We consider the DDM  with $K = 2$ change points at $t \in \{70, 140\}$, yielding $3$ time segments $\{\mathcal{A}_i\}_{i=1}^3$. We set  $\rho_t = 2$ for $t \in \mathcal{A}_1 \cup \mathcal{A}_3$, and $\rho_t = 3$ with reversed layer order for $t \in \mathcal{A}_2$.

\medskip
\noindent\textbf{Scenario 2.} We consider the MSBM with $K=5$ change points at $t \in \{20,60,80,160,180\}$, resulting in $6$ time segments $\{\mathcal{A}_i\}_{i=1}^6$.  We let $\{\mathcal{B}_c(t)\}_{c \in [C_t]}$ be evenly-sized communities and specify the changes as follows: $C_t = 4$ for $t \in \mathcal{A}_1$, $C_t = 2$ for $t \in \mathcal{A}_2$,  $C_t = 4$ for $t \in \mathcal{A}_3$, $C_t = 4$ with reversed layer order for $t \in \mathcal{A}_4$,   $C_t = 3$ for $t \in \mathcal{A}_5$ and $C_t = 4$ for $t \in \mathcal{A}_6$.

\medskip
\noindent\textbf{Scenario 3.}
We consider the MSBM with $K=3$ change points  at $t \in \{50, 100, 150\}$, yielding $4$ time segments $\{\mathcal{A}_i\}_{i=1}^4$.
The number of communities is fixed at $C_t =3$ but in the first layer, the the community sizes vary across segments $(0.3n, 0.4n, 0.3n)$ in $\mathcal{A}_1 \cup \mathcal{A}_4$, $(0.4n, 0.3n, 0.3n)$ in $\mathcal{A}_2$ and $(0.5n, 0.3n, 0.2n)$ in $\mathcal{A}_3$.  The remaining layers retain equal-sized communities.

\medskip
\noindent\textbf{Scenario 4.}
We consider the MSBM with $K = 5$ change points  at $t \in \{20,60,80,160,180\}$, resulting in $6$ time segments $\{\mathcal{A}_i\}_{i=1}^6$. The number of communities is fixed at $C_t =4$ with equal-sized partitions, while the connection probabilities vary across segments. Specifically, for $\epsilon = 0.1$, we let
\[
p_{1,l}  \sim \text{Uniform}\left(0.5 \cdot [0.21 + \delta_t \cdot \epsilon], 0.5 \cdot [0.25 + \delta_t \cdot \epsilon] \right)
\]
and 
\[
p_{2,l} \sim \text{Uniform}\left(0.21 + \delta_t \cdot \epsilon, 0.25 + \delta_t \cdot \epsilon\right),
\]
where $\delta_t = 0$ for $t \in \mathcal{A}_1 \cup \mathcal{A}_5$, $\delta_t = 1$ for $t \in \mathcal{A}_2 \cup \mathcal{A}_4 \cup \mathcal{A}_6$ and $\delta_t = 2$ for $t \in \mathcal{A}_3$.

The changes in \textbf{Scenarios 1} and \textbf{4} follow Model \ref{model-1}, while those in \textbf{Scenarios 2} and \textbf{3} do not, allowing us to assess the robustness of our methods. Table~\ref{tbl:s1s2} presents results for \textbf{Scenarios 1} and \textbf{2}, and Table~\ref{tbl:s3s4} in Appendix~\ref{appendix_syn_data} for \textbf{Scenarios 3} and \textbf{4}.
 Across all scenarios, our method achieves the best overall performance, nearly accurately estimating both the number and locations of change points, and remains robust even when Model~\ref{model-1} is violated.
 For gSeg, Frobenius norm (frob.) inputs yield better results than networks (nets.), while kerSeg performs better with networks, benefiting from its high-dimensional kernel-based design. Although both competitors exhibit low Hausdorff distances $d(\widehat{\mathcal{C}}|\mathcal{C})$, their higher reverse distances $d(\mathcal{C}|\widehat{\mathcal{C}})$ and frequent errors in estimating the number of change points suggest they often detect spurious change points.

Table~\ref{tbl:CI} reports the coverage and average lengths of the confidence intervals constructed via the procedure in Section~\ref{sec:CI} for node size $n \in \{100,150\}$. The proposed method generally achieves strong coverage with reasonably narrow intervals. Coverage is lower in \textbf{Scenario 3}, where violations of Model~\ref{model-1} and relatively small, layer-specific changes pose greater challenges. The performance improves with larger $n$ as the change magnitudes $\kappa_k$ (Model \ref{model-1}) increase.

\begin{table}
\caption{The $95\%$ confidence interval coverage 
(average length) for change points across all scenarios.}
\vspace{1em}
\label{tbl:CI}
\centering
\resizebox{0.7\textwidth}{!}{ 
\begin{tabular}{ l llll }
\toprule
$n$ & Scenario 1 & Scenario 2 & Scenario 3 & Scenario 4 \\
\midrule
$100$ & $100\%$ $(0.003)$ & $100\%$ $(0.106)$ & $76.67\%$ $(1.528)$ & $100\%$ $(0.605)$\\
$150$ & $100\%$ $(0.001)$ & $100\%$ $(0.029)$ & $95.33\%$ $(0.653)$ & $100\%$ $(0.294)$\\
\bottomrule
\end{tabular}
} 
\end{table}

\subsection{Real data experiments}\label{sec:real}

Our analysis incorporates two real data sets, the worldwide agricultural trade network data set presented here and the U.S.~air transport network data set in Appendix~\ref{appendix_real_data}. 

\noindent\textbf{The worldwide agricultural trade network data} are available from \cite{FAOSTAT}. The dataset comprises annual multilayer networks from 1986 to 2020~($T = 35$), with nodes representing countries and layers representing agricultural products. A directed edge within a layer indicates the trade relation between two countries of a specific agricultural product. We use the top $L = 4$ agricultural products by the trade volume and the $n =75$ most active countries based on import/export volume.  
Tuning parameters follow the setup in Section \ref{sec:simulation}, and our method detects change points in  1991, 1999, 2005, and 2013. Results for competing methods and confidence intervals (Section~\ref{sec:CI}) are provided in Appendix~\ref{appendix_real_data}.

The 1991 change point aligns with the German reunification and the dissolution of the Soviet Union, both of which triggered major political shifts that significantly affected the trade dynamics. The 1999 change point corresponds to the World Trade Organization’s (WTO) Third Ministerial Conference,  a key moment in debates on globalization, particularly regarding agricultural subsidies and tariff reductions, with developing nations demanding fairer trade terms. The 2005 change point marks a WTO agreement to eliminate agricultural export subsidies, promoting greater equity in global markets.  Finally, the 2013 change point corresponds to the adoption of the WTO’s Bali Package, the first fully endorsed multilateral agreement, which introduced the Trade Facilitation Agreement and key provisions on food security and tariff quota administration,  significantly impacting agricultural trade.

\section{Conclusion}\label{sec:conclusion}

In this paper, we study offline change point localization and inference in dynamic multilayer networks — a setting that, to the best of our knowledge, has not has not been previously addressed.
We propose a two-stage algorithm and establish its consistency in estimating both the number and locations of change points. 
We further develop local refinement procedures, derive the limiting distributions and introduce a data-driven method for constructing confidence intervals for the true change points.

Several limitations of this work remain open for future research. First, the assumption $\Delta = \Theta(T)$ precludes frequent change points. This could be relaxed using alternative selection strategies such as the narrowest-over-threshold approach  \citep{Baranowski_2019}  instead of greedy selection in this paper. Second, our inference procedure is limited to vanishing jumps. It would be interesting to explore practical procedures for the non-vanishing regime, potentially building on bootstrap methods \citep[e.g.][]{cho2022bootstrap}.  Lastly, the current framework assumes temporal independence; a natural extension would be to incorporate temporal dependence structures  \citep[e.g.][]{padilla2022change, cho2023high}.

\bibliographystyle{plainnat}
\bibliography{references.bib}

\newpage
\appendix
\section*{Appendix}

All technical details of this paper are deferred to the Appendix.  Limiting distributions for the non-vanishing regime are established in Appendix~\ref{sec:add-theory}. Additional algorithms used in our procedures are provided in Appendix~\ref{section:add-algorithm}.
The proof of Theorem~\ref{theorem-2} is given in Appendix~\ref{proof-theorem-2}, while the proofs of the limiting distribution results, including Theorem~\ref{theorem-inference} in the main text, 
and Theorem~\ref{theorem-inference-cont} in Appendix~\ref{sec:add-theory} are presented in Appendix~\ref{proof-theorem-inference}.  Further details and results for Section~\ref{sec:experiment} are collected in Appendix~\ref{sec-add-simulation}.

\section{Limiting distributions in the non-vanishing regime}\label{sec:add-theory}

\begin{theorem}\label{theorem-inference-cont}

Let  $\{\mathbf{A}(t)\}_{t \in [T]},\{\mathbf{A}'(t)\}_{t \in [T]},  \{\mathbf{B}(t)\}_{t \in [T]},  \{\mathbf{B}'(t)\}_{t \in [T]} \subset \{0, 1\}^{n \times n \times L} $ be mutually independent adjacency tensor sequences generated according to Definition~\ref{def-umrdpg-f-dynamic} and satisfying Model~\ref{model-1}, Assumptions~\ref{ass_X_f_Q_f} and \ref{ass-SNR}.  Let $\{ \widehat{\eta}_k \}_{k=1}^{\widetilde{K}}$ be defined in \eqref{def-estimator-final} with  $\{ \widetilde{\eta}_k \}_{k=1}^{\widetilde{K}}$ obtained from Algorithm~\ref{offline-algorithm}, using a threshold $\tau$ satisfying condition stated in Theorem~\ref{theorem-2}.

For $k \in [K]$, if $\kappa_k \to \rho_k$, as $T \to \infty$, with $\rho_k > 0$ being an absolute constant, then  when $ T \to \infty$, we have   $ \vert \widehat{\eta}_k - \eta_k\vert  = O_p(1)$  and 
    \[
    \widehat{\eta}_k - \eta_k \xrightarrow{\mathcal{D}} \argmin_{r \in \mathbb{Z}} \mathcal{P}_k(r), \quad  \mbox{where}  \quad   \mathcal{P}_k(r) =
    \begin{cases}
      - r \rho_k^2 - 2\rho_k\sum_{t=r+1}^0 \langle \boldsymbol{\mathbf{\Psi}}_k ,\mathbf{E}_{k}(t) \rangle, & r < 0, \\
        0, & r = 0, \\
        r\rho_k^2 + 2 \rho_k \sum_{t=1}^r \langle \boldsymbol{\mathbf{\Psi}}_k ,\mathbf{E}_{k+1}(t) \rangle, & r > 0,
    \end{cases}
 \]
    for $r \in \mathbb{Z}$. Here, the normalized jump tensor 
     $\boldsymbol{\mathbf{\Psi}}_k$  is defined in Model~\ref{model-1}, and for any $k \in [K+1]$ and $t \in \mathbb{Z}$,
     $\mathbf{E}_{k}(t) = \mathbf{A}_{k}(t) - \mathbf{P}(\eta_{k})$ with $\{\mathbf{A}_{k}(t)\}_{t \in \mathbb{Z}} \overset{\mathrm{i.i.d.}}{\sim}  \mathrm{MRDPG}(\{X_i\}_{i=1}^{n}, \{W_{(l)}(\eta_{k})\}_{l\in [L]})$.

\end{theorem} 

The proof of Theorem~\ref{theorem-inference-cont} is given in Appendix~\ref{proof-theorem-inference}.

Similar to Theorem~\ref{theorem-inference}, Theorem~\ref{theorem-inference-cont} establishes the uniform tightness $\kappa_k^2 |\widehat{\eta}_k - \eta_k| = O_p(1)$ and further derives the limiting distributions of the refined change point estimators defined in \eqref{def-estimator-final}, which are associated with a two-sided random walk.

\section{Additional algorithms}\label{section:add-algorithm}
We present the tensor heteroskedastic principal component analysis (TH-PCA) algorithm introduced in \cite{han2022optimal}, incorporating an additional truncation step, in Algorithm~\ref{thpca}. Its subroutine, 
the heteroskedastic principal component analysis (H-PCA) algorithm proposed by  \cite{zhang2022heteroskedastic}, is provided in Algorithm~\ref{hpca}.

\begin{algorithm}
\caption{Tensor heteroskedastic principal component analysis, TH-PCA$(\mathbf{A}, (r_1, r_2, r_3), \tau_1, \tau_2)$} \label{thpca}
\begin{algorithmic}
    \INPUT{Tensor $\mathbf{A} \in \mathbb{R}^{p_1 \times p_2 \times p_3}$, ranks $r_1, r_2, r_3 \in \mathbb{N}^+$, thresholds $\tau_1, \tau_2 \geq 0$}
    \For{$s \in [3]$} 
        \State{$
        \widehat{U}_s \leftarrow \mbox{H-PCA}( \mathcal{M}_s (\mathbf{A})  \mathcal{M}_s (\mathbf{A})^{\top}, r_s)$ \Comment{{\small See Algorithm~\ref{hpca} for H-PCA and Section \ref{sec:notation} for $\mathcal{M}_s (\mathbf{A})$}}}
    \EndFor
    \State{$
    \widetilde{\mathbf{P}} \leftarrow \mathbf{A} \times_1 \widehat{U}_1 \widehat{U}_1^{\top} \times_2 \widehat{U}_2 \widehat{U}_2^{\top} \times_3 \widehat{U}_3  \widehat{U}_3^{\top}$ \Comment{{\small See Section \ref{sec:notation} for $\times_s$}}}
    \For{$\{i, j, l\} \in [p_1] \times [p_2] \times [p_3]$}
        \State{$\widehat{\mathbf{P}}_{i,j,l} \leftarrow \min\big\{\tau_1, \max\{-\tau_2, \widetilde{\mathbf{P}}_{i,j,l}\} \big\}$}
    \EndFor
    \OUTPUT{$\widehat{\mathbf{P}} \in \mathbb{R}^{p_1 \times p_2 \times p_3}$}
\end{algorithmic}
\end{algorithm}

\begin{algorithm}
\caption{Heteroskedastic principal component analysis, H-PCA$(\Sigma, r)$} \label{hpca}
\begin{algorithmic}
    \INPUT{Matrix $\Sigma \in \mathbb{R}^{n \times n}$, rank $r \in \mathbb{N}^+$.}
    \Initialise{ $\widehat{\Sigma}^{(0)} \leftarrow \Sigma$, $\mathrm{diag}\big(\widehat{\Sigma}^{(0)}\big) \leftarrow 0$, $T \leftarrow 5\log\{\sigma_{\min}(\Sigma)/n\}$}
    \For{$t \in \{0\} \cup [T-1]$}
        \State{Singular value decomposition
        $\widehat{\Sigma}^{(t)} = \sum_{i=1}^n\sigma^{i, (t)} \mathbf{u}^{i,(t)} ( \mathbf{v}^{i, (t)})^{\top},  \quad  \sigma^{1, (t)} \geq \cdots \geq  \sigma^{n, (t)}  \geq 0 
        $}
        \State{$
        \widetilde{\Sigma}^{(t)} \leftarrow \sum_{i = 1}^r \sigma^{i, (t)} \mathbf{u}^{i, (t)} \big(\mathbf{v}^{i, (t)}\big)^{\top}
        $}
        \State{$\widehat{\Sigma}^{(t+1)} \leftarrow \widehat{\Sigma}^{(t)}$,   $\mathrm{diag}\big(\widehat{\Sigma}^{(t+1)}\big) \leftarrow \mathrm{diag}\big(\widetilde{\Sigma}^{(t)}\big)
        $}
    \EndFor
    \State{$U \leftarrow (\mathbf{u}^1, \ldots, \mathbf{u}^r)$ from top-$r$ left singular vectors of $\widehat{\Sigma}^{(T)}$}
    \OUTPUT{$U \in \R^{n \times r}$}
\end{algorithmic}
\end{algorithm}

\section{Proof of Theorem~\ref{theorem-2}}\label{proof-theorem-2} 

The proof of Theorem~\ref{theorem-2} is in Appendix~\ref{sec:proof-theorem-2} with all necessary auxiliary results in Appendix~\ref{sec:proof-auz}.

\subsection{Proof of Theorem~\ref{theorem-2}}\label{sec:proof-theorem-2}

\begin{proof}
We first define the event
\[
\mathcal{A}   = \Big\{ \widetilde{K} = K \mbox{ and } \vert b_k - \eta_k \vert  \leq \tilde{\epsilon}, \, \forall k \in [K] \Big\}, \quad \mbox{where }  \tilde{\epsilon} = C_{\tilde{\epsilon}}       \log(T) \bigg\{ \frac{n \sqrt{L} \log^{1/2}(T)}{\kappa^2} + \frac{\sqrt{\Delta}}{\kappa}  \bigg\},
\]
where $\{b_k\}_{k=1}^{\widetilde{K}}$ are preliminary change point estimates obtained from \textbf{Stage I} in Algorithm~\ref{offline-algorithm}. 
Then by Proposition~\ref{theorem-1}, it holds that  
\[
\mathbb{P}\{ \mathcal{A}\} \geq 1 - C_0T^{-c_0}
\]
and $C_{\tilde{\epsilon}}, C_0, c_0 >0$ are absolute constants. Since $\{\mathbf{A}'(t)\}_{t=1}^T \cup \{\mathbf{B}'(t)\}_{t=1}^T$ are independent of $\{\mathbf{A}(t)\}_{t=1}^T \cup \{\mathbf{B}(t)\}_{t=1}^T$, the distribution of $\{\mathbf{A}'(t)\}_{t=1}^T \cup \{\mathbf{B}'(t)\}_{t=1}^T$  remains unaffected under the conditioning on the event $\mathcal{A}$. All subsequent analysis in this proof is carried out under the event~$\mathcal{A}$. Consequently, we can derive that 
\begin{equation}\label{theorem-2-nu_k}
    \vert b_k - \eta_{k} \vert \leq \tilde{\epsilon} \leq \Delta/6,  \quad \forall k \in [K],
\end{equation}
where the last inequality follows from  Assumption~\ref{ass-SNR} and the fact that $C_{\mathrm{SNR}}$ is a sufficiently large constant.

\medskip
\noindent \textbf{Step 1.}  We first establish that for any $k \in [K]$,  each working interval $(s_k, e_k)$ contains exactly one true change point, namely $\eta_k$, and the two endpoints are well separated. 

From \eqref{theorem-2-nu_k}, we obtain that
$ \eta_k \in [b_{k-1}, b_{k+1}]$, 
\[
\eta_k - b_{k-1} \geq
  \eta_k - \eta_{k-1} - \vert \eta_{k-1} - b_{k-1}\vert 
   \geq  \Delta - \Delta/6    \geq 5\Delta/6,
\]
and 
\[
b_{k+1} - \eta_k \geq  \eta_{k+1} - \eta_k - \vert \eta_{k+1} - b_{k+1} \vert \geq  \Delta - \Delta/6 \geq 5\Delta/6.
\]
Similarly, we can derive that
\[
\min\{b_k - b_{k-1}, \,b_{k+1} - b_k\} \geq  2\Delta/3.
\]
As a result, the working interval 
\[ (s_k,e_k] =  \big(b_{k-1}  +(b_k - b_{k-1})/2,\;
   b_{k+1} - (b_{k+1} - b_k)/2 \big],
\]
contains exactly one change point $\eta_k$. 
For any $t \in (s_k, e_k)$,  denote $m^{s_k, e_k}_t = \mathrm{rank}(\widetilde{Q}^{s_k, e_k}(t))$ with $\widetilde{Q}^{s_k, e_k}(t)$ defined in \eqref{def-tilde-Q}. Then we have that  
\begin{equation}\label{theorem-2-eq-2}
    m^{s_k, e_k}_t \leq m_k + m_{k+1}\leq 2 m_{\max},
\end{equation}
where $m_{k} = \mathrm{rank}\big(Q(\eta_k)\big)$ with $Q(\eta_k)$ defined in \eqref{def-Q},
and $m_{\max} = \max_{k \in [K+1]} m_k$.

In addition, we  have that 
\[
b_k - s_k =  (b_k - b_{k-1})/2\geq  \Delta/3,
\]
and
\[
e_k - b_k  =  (b_{k+1} - b_{k})/2  \geq  \Delta/3.
\]
Therefore,
\begin{equation}\label{theorem-2-eq-1}
\min \{e_k - b_k,b_k - s_k\} \geq  \Delta/3.
\end{equation}

\medskip
\noindent \textbf{Step 2.} We now show that the population statistics $\widetilde{\mathbf{P}}^{s_k,e_k}(b_k)$ provide a sufficiently strong signal within each working interval $(s_k, e_k]$.

By Lemma~\ref{lemma-s-15}, it holds that
\[
\big\|\widetilde{\mathbf{P}}^{s_k,e_k}(t)\big\|_{\mathrm{F}}^2 
=
\begin{cases}
\frac{t - s_k}{(e_k - s_k)(e_k - t)}(e_k- \eta_k)^2 
\kappa_k^2,
& s_k <t \le \eta_k, \\
\frac{e_k - t}{(e_k - s_k)(t - s_k)}(\eta_k - s_k)^2 
\kappa_k^2,
& \eta_k <  t <e_k.
\end{cases}
\]
Define the scaling factor
\[
\widetilde{\Delta}_k = 
\sqrt{\frac{(b_k - s_k)\,(e_k - b_k)}{e_k - s_k}}.
\]
Assuming without loss of generality that  $b_k \le \eta_k$ and using \eqref{theorem-2-eq-1}, we obtain
\begin{equation}\label{eq2-step2.1.0}
\widetilde{\Delta}_k^2 
\geq \frac{\min\{b_k - s_k,\, e_k - b_k\}}{2} 
\geq  \frac{\Delta}{6}.
\end{equation}
Thus, we have that
\begin{align}\label{eq2-step2.1}
\big\|\widetilde{\mathbf{P}}^{s_k,e_k}(b_k)\big\|_{\mathrm{F}}^2
= & 
\frac{b_k - s_k}{(e_k - s_k)(e_k - b_k)}
(e_k - \eta_k)^2 \ \kappa_k^2 
= \widetilde{\Delta}_k^2 
\bigg(\frac{e_k - \eta_k}{e_k - b_k}\bigg)^{2}
\kappa_k^2 \nonumber\\
=  &  \widetilde{\Delta}_k^2 
\bigg(1 - \frac{\eta_k - b_k}{e_k - b_k}\bigg)^{2} \kappa_k^2 \geq  \frac{\Delta}{6} (1 - \frac{\Delta/6}{\Delta/3})^{2} \kappa_k^2 =  \Delta\kappa_k^2/24,
\end{align}
where the first inequality follows from \eqref{theorem-2-nu_k}, \eqref{theorem-2-eq-1} and \eqref{eq2-step2.1.0}.

\medskip
\noindent 
\textbf{Step 3.} 
Note that each entry of the tensor $\widetilde{\mathbf{B}'}^{s, e}(b_k)$ is independently $c_{\sigma}$-sub-Gaussian distributed with mean tensor $\E\{ \widetilde{\mathbf{B}'}^{s, e}(b_k) \}= \widetilde{\mathbf{P}}^{s_k,e_k}(\eta_k)$  and an absolute constant $c_{\sigma} >0$. By Theorem 4 and Lemma 5 in \cite{wang2023multilayer}, and Assumption~\ref{ass_X_f_Q_f} $(i)$ and $(ii)$, for any $t \in (s, e)$, it holds that
\[
\mathbb{P} \bigg\{
\big\| \widehat{\mathbf{P}}^{s_k,e_k}(t) - \widetilde{\mathbf{P}}^{s_k,e_k}(t) \big\|_F \leq C_{1} \sqrt{(d^2m^{s, e}_t + nd+Lm^{s, e}_t )   \log(T)} \bigg\}  \geq 1 - T^{-c_{1}},
\]
for some constants $C_{1} >0$ and $c_1 >3$. By \eqref{theorem-2-eq-2}, we can derive that
\[
\mathbb{P} \bigg\{
\big\| \widehat{\mathbf{P}}^{s_k,e_k}(t) - \widetilde{\mathbf{P}}^{s_k,e_k}(t) \big\|_F \leq C_{2} \sqrt{( d^2m_{\max} + nd+Lm_{\max})   \log(T)} \bigg\}  \geq 1 - T^{-c_{1}},
\]
where $C_{2} >0$ is a constant. Define the event
\begin{align}\label{def-event-B}
\mathcal{B} = \Bigg\{
\sup_{\substack{0 \leq s_k < t < e_k \leq T \\    (s_k,e_k) \mbox{ contains only one change point } \eta_k }} & \big\| \widehat{\mathbf{P}}^{s_k,e_k}(t) - \widetilde{\mathbf{P}}^{s_k,e_k}(t) \big\|_F \leq \nonumber\\
& C_{\mathcal{B}} \sqrt{(d^2m_{\max} + nd+Lm_{\max} )   \log(T)} \Bigg\}.
\end{align}
with a constant $C_{\mathcal{B}}>0$.
By the union bound argument, it holds that 
\[
\mathbb{P} \big\{ \mathcal{B} \big\}  \geq 1 - T^{- c_{\mathcal{B}}},
\]
with a constant $c_{\mathcal{B}}>0$.
By the event $\mathcal{B}$ and the triangle equality, we have that 
\begin{equation}\label{eq2-step3.1}
\big\| \widehat{\mathbf{P}}^{s_k,e_k}(b_k) \big\|_{\mathrm{F}}  \geq  \big\| \widetilde{\mathbf{P}}^{s_k,e_k}(b_k) \big\|_{\mathrm{F}}  - 
C_{\mathcal{B}} \sqrt{(d^2m_{\max} + nd+Lm_{\max} )   \log(T)}  \geq  \kappa_{k}\sqrt{\Delta}/48,
\end{equation}
where the last inequality follows from  \eqref{eq2-step2.1},  Assumption~\ref{ass-SNR} and the fact that $C_{\mathrm{SNR}}$ is a sufficiently large constant. 
As a consequence,
\begin{align*}
&  2 \Big\langle
\widetilde{\mathbf{P}}^{s_k,e_k}(b_k) / \|\widetilde{\mathbf{P}}^{s_k,e_k}(b_k)\|_{\mathrm{F}} ,
\widehat{\mathbf{P}}^{s_k,e_k}(b_k) / \|\widehat{\mathbf{P}}^{s_k,e_k}(b_k)\|_{\mathrm{F}} 
\Big \rangle \\
= &   2 -  \Bigg\|  \frac{
\widetilde{\mathbf{P}}^{s_k,e_k}(b_k)}{\|\widetilde{\mathbf{P}}^{s_k,e_k}(b_k)\|_{\mathrm{F}}} - \frac{ \widehat{\mathbf{P}}^{s_k,e_k}(b_k) 
}{\|\widehat{\mathbf{P}}^{s_k,e_k}(b_k)\|_{\mathrm{F}} }\Bigg\|_{\mathrm{F}}^2 \\
= &  2 -  \Bigg\|  \frac{
 \big( \widetilde{\mathbf{P}}^{s_k,e_k}(b_k) - \widehat{\mathbf{P}}^{s_k,e_k}(b_k) \big)  \|\widetilde{\mathbf{P}}^{s_k,e_k}(b_k)\|_{\mathrm{F}}  }{\|\widetilde{\mathbf{P}}^{s_k,e_k}(b_k)\|_{\mathrm{F}} \|\widehat{\mathbf{P}}^{s_k,e_k}(b_k)\|_{\mathrm{F}} } \\
 & \hspace{0.5cm}  +   \frac{
  \widetilde{\mathbf{P}}^{s_k,e_k}(b_k)  \big( \|\widehat{\mathbf{P}}^{s_k,e_k}(b_k)\|_{\mathrm{F}} - \|\widetilde{\mathbf{P}}^{s_k,e_k} (b_k)\|_{\mathrm{F}}\big) }{\|\widetilde{\mathbf{P}}^{s_k,e_k}(b_k)\|_{\mathrm{F}} \|\widehat{\mathbf{P}}^{s_k,e_k}(b_k)\|_{\mathrm{F}} }\Bigg\|_{\mathrm{F}}^2 \\
 \geq  &  2 -  \Bigg( \frac{
\| \widetilde{\mathbf{P}}^{s_k,e_k}(b_k) - \widehat{\mathbf{P}}^{s_k,e_k}(b_k) \|_{\mathrm{F}} \|\widetilde{\mathbf{P}}^{s_k,e_k}(b_k)\|_{\mathrm{F}} }{\|\widetilde{\mathbf{P}}^{s_k,e_k}(b_k)\|_{\mathrm{F}} \|\widehat{\mathbf{P}}^{s_k,e_k}(b_k)\|_{\mathrm{F}} }\\
& \hspace{0.5cm}  + \frac{
\| \widetilde{\mathbf{P}}^{s_k,e_k}(b_k)\|_{\mathrm{F}} \|\widehat{\mathbf{P}}^{s_k,e_k}(b_k) - \widetilde{\mathbf{P}}^{s_k,e_k} (b_k)\|_{\mathrm{F}}}{\|\widetilde{\mathbf{P}}^{s_k,e_k}(b_k)\|_{\mathrm{F}} \|\widehat{\mathbf{P}}^{s_k,e_k}(b_k)\|_{\mathrm{F}} } \Bigg)^2\\
=  &  2 -  \Bigg( \frac{
\big( \| \widetilde{\mathbf{P}}^{s_k,e_k}(b_k)\|_{\mathrm{F}} + \|\widetilde{\mathbf{P}}^{s_k,e_k}(b_k)\|_{\mathrm{F}}  \big) \|\widehat{\mathbf{P}}^{s_k,e_k}(b_k) - \widetilde{\mathbf{P}}^{s_k,e_k} (b_k)\|_{\mathrm{F}}}{\|\widetilde{\mathbf{P}}^{s_k,e_k}(b_k)\|_{\mathrm{F}} \|\widehat{\mathbf{P}}^{s_k,e_k}(b_k)\|_{\mathrm{F}} } \Bigg)^2\\
\geq &   2 - 4 \Bigg( \frac{
\|\widetilde{\mathbf{P}}^{s_k,e_k}(b_k) -  \widehat{\mathbf{P}}^{s_k,e_k}(b_k)\|_{\mathrm{F}} 
}{
\min\{\|\widetilde{\mathbf{P}}^{s_k,e_k}(b_k)\|_{\mathrm{F}}, \|\widehat{\mathbf{P}}^{s_k,e_k}(b_k)\|_{\mathrm{F}} \}
} \Bigg)^2 \\
\geq & 
2 - 4 \frac{48^2 C_{\mathcal{B}}^2 (d^2m^{s, e}_t + nd+Lm^{s, e}_t )   \log(T)}{\kappa_{k}^2 \Delta} \geq 1,
\end{align*}
where the first inequality follows from the reverse triangle inequality, the third inequality follows from
the definition of the event $\mathcal{B}$, \eqref{eq2-step2.1} and \eqref{eq2-step3.1},
and the final inequality follows from Assumption~\ref{ass-SNR} and the fact that $C_{\mathrm{SNR}}$ is a sufficiently large constant. Therefore,
\begin{equation}\label{eq2-step3.2}
\Big\langle
\widetilde{\mathbf{P}}^{s_k,e_k}(b_k),
\widehat{\mathbf{P}}^{s_k,e_k}(b_k) / \|\widehat{\mathbf{P}}^{s_k,e_k}(b_k)\|_{\mathrm{F}} 
\Big\rangle \geq \|\widetilde{\mathbf{P}}^{s_k,e_k}(b_k)\|_{\mathrm{F}} /2 \geq \kappa_k\sqrt{\Delta/96}, 
\end{equation}
where the last inequality follows from \eqref{eq2-step2.1}.

\medskip
\noindent 
\textbf{Step 4.} Since $\{\mathbf{A}'(t)\}_{t=1}^T$ is independent of $\{\mathbf{B}'(t)\}_{t=1}^T$, the distribution of $\{\mathbf{A}'(t)\}_{t=1}^T$ remain unaffected under the conditioning on the event $\mathcal{B}$. By the truncation in the construction of $\widehat{\mathbf{P}}^{s_k,e_k}(b_k )$ stated in Algorithm~\ref{thpca}, 
we have that
\[
\|\widehat{\mathbf{P}}^{s_k,e_k}(b_k )\|_{\infty}  \leq   \sqrt{\frac{(e_k-b_k) (b_k-s_k)}{(e_k-s_k)} }
\]
Combined with \eqref{eq2-step3.1}, it follows that
\[
(e_k-s_k)^{-1/2} \|\widehat{\mathbf{P}}^{s_k,e_k}(b_k )\|_{\infty}/ \|\widehat{\mathbf{P}}^{s_k,e_k}(b_k )\|_{\mathrm{F}} \leq\frac{1}{ c_3 \kappa_k \sqrt{\Delta} },
\]
for some constant $c_3 >0$.  Applying Lemma \ref{lemma-s-2}, we obtain for any $\varepsilon > 0$
\begin{align*}
& \mathbb{P} \Big\{ \Big\vert \frac{1}{\sqrt{e_k-s_k}} \sum_{t=s_k+1}^{e_k} \Big\langle \mathbf{P}(t) - \mathbf{A}'(t), \widehat{\mathbf{P}}^{s_k,e_k}(b_k )/ \|\widehat{\mathbf{P}}^{s_k,e_k}(b_k )\|_{\mathrm{F}} \Big\rangle \Big\vert \geq \varepsilon \Big\} \\
\leq & 2 \exp \Bigg\{ -c_4 \frac{\varepsilon^2}{1  + \varepsilon  / (c_3 \kappa_k \sqrt{\Delta})} \Bigg\}.
\end{align*}
where $c_4 >0$ is a constant.
Choosing $\varepsilon = C_3 \sqrt{\log(T)}$ for a large enough constant $C_3 >0$, and applying Assumption~\ref{ass-SNR} and   the fact that $C_{\mathrm {SNR}}$ is a sufficiently large constant,  we finally derive that
\begin{equation}\label{eq2-step4.1}
\mathbb{P} \Big\{ \Big \vert \frac{1}{\sqrt{e_k-s_k}} \sum_{t=s_k+1}^{e_k} \Big\langle \mathbf{P}(t) - \mathbf{A}'(t), \widehat{\mathbf{P}}^{s_k,e_k}(b_k )/ \|\widehat{\mathbf{P}}^{s_k,e_k}(b_k )\|_{\mathrm{F}}  \Big\rangle \Big\vert \geq C_3 \sqrt{ \log(T)} \Big\} 
\leq 2 T^{-c_5},
\end{equation}
where $c_5 >3$ is a constant.
A similar argument also demonstrates that
\begin{equation}\label{eq2-step4.2}
\mathbb{P} \Big\{ \Big\vert \Big\langle \widetilde{\mathbf{P}}^{s_k,e_k}(t) - \widetilde{\mathbf{A}'}^{s_k,e_k}(t), \widehat{\mathbf{P}}^{s_k,e_k}(b_k ) / \|\widehat{\mathbf{P}}^{s_k,e_k}(b_k )\|_{\mathrm{F}}   \Big\rangle \Big\vert \geq C_3 \sqrt{\log(T)} \Bigg\} 
\leq 2 T^{-c_5}.
\end{equation}

\medskip
\noindent
\textbf{Step 5.} We now consider the univariate time series defined for all $t \in (s_k, e_k)$ as
\[
y(t) = \Big\langle \mathbf{A}'(t), \widehat{\mathbf{P}}^{s_k,e_k}(b_k ) /\|\widehat{\mathbf{P}}^{s_k,e_k}(b_k )\|_{\mathrm{F}} \Big \rangle
\]
and 
\[
y^{s_k, e_k}(t) = \Big\langle \widetilde{\mathbf{A}'}^{s_k, e_k}(t), \widehat{\mathbf{P}}^{s_k,e_k}(b_k ) /\|\widehat{\mathbf{P}}^{s_k,e_k}(b_k )\|_{\mathrm{F}} \Big \rangle.
\]
Conditional on the event $\mathcal{B}$, define the corresponding mean functions
\[
 f(t) = \mathbb{E}(y(t)) = \Big \langle \mathbf{P}(t), \widehat{\mathbf{P}}^{s_k,e_k}(b_k ) /\|\widehat{\mathbf{P}}^{s_k,e_k}(b_k )\|_{\mathrm{F}} \Big\rangle,
 \]
and
\[
f^{s_k, e_k}(t) = \mathbb{E}(y^{s_k, e_k}(t)) = \Big\langle \widetilde{\mathbf{P}}^{s_k, e_k}(t), \widehat{\mathbf{P}}^{s_k,e_k}(b_k ) /\|\widehat{\mathbf{P}}^{s_k,e_k}(b_k )\|_{\mathrm{F}} \Big \rangle.
\]
The function $f(t)$ is a piecewise constant on $(s_k, e_k]$ with a single change point at $\eta_k$.   Using \eqref{eq2-step3.2}, we obtain that
\[
\vert f^{s_k,e_k}(\eta_k) \vert  
\Big\vert \geq \kappa_k \sqrt{\Delta /96},
\]
Moreover, from \eqref{eq2-step4.1}, \eqref{eq2-step4.2} and an union bound argument, we have that
\begin{align*}
\mathbb{P} \Bigg\{ \sup_{\substack{0 \leq s_k < t < e_k \leq T \\   (s_k,e_k) \mbox{ contains only one change point } \eta_k }} &  \bigg| \frac{1}{\sqrt{e_k-s_k}} \sum_{t=s_k+1}^{e_k} \left( y(t) - f(t) \right) \bigg| \\
& \hspace{0.5cm} \geq C_3 \sqrt{ \log(T)} \Bigg\} \leq 2T^{-c_6}
\end{align*}
and
\[
\mathbb{P} \Bigg\{ \sup_{\substack{0 \leq s_k < t < e_k \leq T \\  (s_k,e_k) \mbox{ contains only one change point } \eta_k }}  \big\vert y^{s_k,e_k}(t) - f^{s_k,e_k}(t) \big\vert \geq C_3 \sqrt{ \log(T)} \Bigg\} \leq 2T^{-c_6},
\]
for some constant $c_6 > 0$.  Applying Lemma 12 from \cite{wang2017optimal} with $\lambda = C \sqrt{\log(T)}$, it follows that the estimated change point $\widetilde{\eta}_k = \argmax_{s_k < t < e_k} \vert y^{s_k,e_k}(t)\vert $ is an undetected change point and satisfies for a large enough constant $C_5 > 0$,
\[
\vert \widetilde{\eta}_k - \eta_k\vert  \leq C_5 \frac{\log (T)}{\kappa_k^2},
\]
which completes the proof. 

\end{proof}

\subsection{Additional results}\label{sec:proof-auz}

\begin{proposition}\label{theorem-1}
Let $\{ b_k \}_{k=1}^{\widetilde{K}}$ denote the output of \textbf{Stage I} in Algorithm~\ref{offline-algorithm} applied to two independent adjacency tensor sequences $\{\mathbf{A}(t)\}_{t \in [T]}, \{\mathbf{B}(t)\}_{t \in [T]} \subset \{0, 1\}^{n \times n \times L} $, generated according to  Definition~\ref{def-umrdpg-f-dynamic} and satisfying Model~\ref{model-1} and Assumption~\ref{ass-SNR}. Suppose the threshold $\tau$ is chosen such that 
\begin{equation}\label{eq-tau-def-thm_f_new}
   c_{\tau,1}   n\sqrt{L} \log^{3/2}(T) <   \tau  < c_{\tau,2} \kappa^2 \Delta,
\end{equation}
where $c_{\tau,1}, c_{\tau,2} > 0$ are sufficiently large and small absolute constants, respectively.

Then, it holds that 
\[
\mathbb{P} \Big\{ \widetilde{K} = K \mbox{ and } \vert b_k - \eta_k \vert  \leq \tilde{\epsilon}, \, \forall k \in [K] \Big\} \geq 1 - CT^{-c},
\]
where 
\[
 \tilde{\epsilon} = C_{\tilde{\epsilon}}       \log(T) \bigg\{ \frac{n \sqrt{L} \log^{1/2}(T)}{\kappa^2} + \frac{\sqrt{\Delta}}{\kappa}  \bigg\},
\]
and $C_{\tilde{\epsilon}}, C, c >0$ are absolute constants. 
\end{proposition}

\begin{proof}
The proof presented here is a minor modification of Theorem 1 in \cite{wang2021optimal}. For completeness, we include the full details below.

For $0 \leq s < t < e \leq T $, we define the event
\[
\mathcal{A}(s, t, e) = \Big\{
\big| 
\big\langle \widetilde{\mathbf{A}}^{s,e}(t), \widetilde{\mathbf{B}}^{s,e}(t) \big\rangle - \|\widetilde{\mathbf{P}}^{s,e}(t)\|_{\mathrm{F}}^2 
\big| 
\leq C_{\mathcal{A}} \log(T) \big( \|\widetilde{\mathbf{P}}^{s,e}(t)\|_{\mathrm{F}} + \log^{1/2}(T) n\sqrt{L} \big) 
\Big\},
\]
where $\widetilde{\mathbf{P}}^{s,e}(t)$ is defined in \eqref{def-average} and $C_{\mathcal{A}} > 0 $ is a constant.
Due to Lemma S.4 in \cite{wang2021optimal}, it holds that $ \mathbb{P}(\mathcal{A}(s,t,e)^c) \leq C_1 T^{-c_1}$ for some constants $C_1>0$ and $c_1 > 3$. By an union bound argument, it holds that 
\[
\mathbb{P}(\mathcal{A}) = \mathbb{P} \bigg\{ \bigcup_{0 \leq s < t < e \leq T} \mathcal{A}(s,t,e) \bigg\} \geq 1 -  C_1 T^{3-c_1}.
\]
All subsequent analysis in this proof is carried out under the event $\mathcal{A} $.

We aim to demonstrate that, conditioned on the event $\mathcal{A}$ and assuming that the algorithm has accurately detected and localized change points so far, the procedure will also successfully identify any remaining undetected change point, if one exists,  and estimate its location within an error of $\tilde{\epsilon} $. To this end, it suffices to consider an arbitrary time interval $ 0 \leq s  < e  \leq  T$ that satisfies
\[
\eta_{r-1} \leq s < \eta_r < \dots < \eta_{r+q} < e \leq  \eta_{r+q+1}, \quad q \geq -1,
\]
and
\[
\max \big\{ \min \{\eta_r - s, s - \eta_{r-1} \}, \min \{\eta_{r+q+1} - e, e - \eta_{r+q} \} \big\} \leq \tilde{\epsilon},
\]
where $q = -1$ indicates that there is no change point contained in $(s,e)$ and 
\[ 
\tilde{\epsilon} = C_{\tilde{\epsilon}}       \log(T) \bigg\{ \frac{n \sqrt{L} \log^{1/2}(T)}{\kappa^2} + \frac{\sqrt{\Delta}}{\kappa}  \bigg\}.
\]
with an absolute constant $ C_{\tilde{\epsilon}} >0$.
By Assumption~\ref{ass-SNR}, we have that 
\begin{equation}\label{eq-epsilon}
\tilde{\epsilon} \leq C_{\tilde{\epsilon}} \Delta \bigg( \frac{1}{C_{\mathrm{SNR}}^2 \log^{1/2 }(T)} + \frac{1}{C_{\mathrm{SNR}} \sqrt{nL^{1/2}}} \bigg) \leq \Delta/64,
\end{equation}
where the final inequality follows that $C_{\mathrm{SNR}}$ is large enough. Consequently, for any change point~$\eta_k$, it must be that either $\vert \eta_k - s \vert \leq \tilde{\epsilon} $ or $\vert \eta_k - s  \vert \geq \Delta - \tilde{\epsilon}  \geq 3\Delta/4$. This implies that if
$\min\{ \vert e -\eta_k \vert, \vert\eta_k - s\vert \} \leq \tilde{\epsilon}$,  then $\eta_k$ corresponds to a change point that has already been detected and estimated within an error of at most $\tilde{\epsilon}$ during the previous induction step. We refer to a change point $ \eta_k $ in $(s,e)$ as undetected if 
\[
\min \{ \eta_k - s, \eta_k - e \} \geq 3\Delta/4.
\]

To complete the induction step, it suffices to show that $ \mathrm{SBS}\big( (s, e), \tau,\mathcal{J} \big)$  satisfies the following tow properties:
$(i)$ It does not detect any new change point within $(s,e)$ if all change points in that interval have already been detected; and 
$(ii)$ It detects a point  $b$ in $(s,e)$ such that $\vert\eta_k - b\vert \leq \tilde{\epsilon} $ if there exists at least one previously undetected change point $\eta_k$ in $(s,e)$.

\medskip \noindent
\textbf{Step 1.} Assume that there are no undetected change points within the interval $(s,e)$. Then, for any interval $  (\alpha', \beta'] \in  \mathcal{J}$ with $ (\alpha', \beta'] \subseteq (s,e]$, one of the following scenarios must hold:
$(i)$ The interval $(\alpha', \beta')$ contains no change points;
$(ii)$ The interval $(\alpha', \beta')$ contains exactly one change point $\eta_k$ and $\min\{\eta_k - \alpha', \beta' - \eta_k\} \leq \tilde{\epsilon}$; 
$(iii)$ The interval  $(\alpha', \beta')$ contains two change points  $\eta_k$ and $\eta_{k+1}$, and $\max\{\eta_k - \alpha', \beta' - \eta_{k+1}\} \leq \tilde{\epsilon}$.

We focus on analyzing the scenario $(iii)$, as the other two scenarios are similar and more straightforward. 
If scenario $(iii)$ holds, then by \eqref{eq-epsilon}, we have
\[
\tilde{\epsilon} \leq \Delta/64 \leq  (\beta' - \alpha')/64,
\]
This implies that the interval
\[
(\alpha, \beta] = \big(\alpha' + 64^{-1}(\beta' - \alpha') ,\beta' - 64^{-1}(\beta' - \alpha') \big],
\]
contains no change points. Note that $\widetilde{\mathbf{P}}^{\alpha, \beta}(t) = 0$ for all $t \in (\alpha, \beta)$, since there are no true change points within $(\alpha, \beta)$. Moreover, by the event $\mathcal{A}$, 
\[
\max_{\alpha < t < \beta} \big \langle \widetilde{\mathbf{A}}^{\alpha, \beta}(t), \widetilde{\mathbf{B}}^{\alpha, \beta}(t) \big \rangle \leq  C_{\mathcal{A}}  n\sqrt{L} \log^{3/2}(T).
\]
Therefore, if the input parameter $\tau$ satisfies
\[
  \tau >  C_{\mathcal{A}}  n\sqrt{L} \log^{3/2}(T),
\]
Algorithm~\ref{offline-algorithm} will correctly reject the existence of undetected change points.

\medskip \noindent
\textbf{Step 2.} Now suppose there exists a change point $\eta_k \in (s, e)$ such that 
\[
\min\{\eta_k - s, \eta_k - e\} \geq 3\Delta/4.
\]
Let $a_{\mathcal{I}}, b_{\mathcal{I}}$ and $\mathcal{I}^*$ be as defined in the procedure $\mathrm{SBS}\big((s, e), \tau, \mathcal{J} \big)$.  Denote  $\mathcal{I}^* = (\alpha^{* \prime}, \beta^{* \prime}]$. By Lemma 8 in \cite{madrid2022change}, for any change point $\eta_k \in (s,e)$ satisfying $\min\{\eta_k - s, e - \eta_k\} \geq 3\Delta/4$, 
there exists an interval $(\alpha', \beta'] \subseteq (s,e]$ containing only one $\eta_k$ such that
\[
\eta_k - 3\Delta/4 \leq \alpha' \leq \eta_k - 3 \Delta/16 \quad \mbox{and} \quad \eta_k + 3\Delta/16 \leq \beta' \leq \eta_k + 3\Delta/4.
\]
Since $(\alpha, \beta] = [\alpha' + (\beta' - \alpha')/64, \beta' - (\beta' - \alpha')/64]$, we have 
\[
\eta_k - \Delta3/4 \leq \alpha \leq \eta_k - \Delta/8 \quad \mbox{and} \quad \eta_k + \Delta/8 \leq \beta \leq \eta_k + \Delta3/4.
\]
On the event $\mathcal{A}$, it holds that
\begin{equation}\label{theorem1_eq1}
\big\langle \widetilde{\mathbf{A}}^{\alpha,\beta}(\eta_k), \widetilde{\mathbf{B}}^{\alpha,\beta}(\eta_k) \big\rangle \geq \|\widetilde{\mathbf{P}}^{\alpha,\beta}(\eta_k)\|^2_{\mathrm{F}} -  C_{\mathcal{A}} \log(T) \big\{\log^{1/2}(T)n \sqrt{L} + \|\widetilde{\mathbf{P}}^{\alpha,\beta}(\eta_k)\|_{\mathrm{F}} \big\}.
\end{equation}
Furthermore, by Lemma~\ref{lemma-s-15}, it hold that 
\[
\|\widetilde{\mathbf{P}}^{\alpha,\beta}(\eta_k)\|^2_{\mathrm{F}} = \frac{(\eta_k - \alpha)(\beta - \eta_k)}{\beta - \alpha} \kappa_k^2. 
\]
Then we can derive that 
\begin{equation}\label{theorem1_eq2}
\|\widetilde{\mathbf{P}}^{\alpha,\beta}(\eta_k)\|^2_{\mathrm{F}} \geq \kappa_k^2\Delta/16  \quad \mbox{and} \quad \|\widetilde{\mathbf{P}}^{\alpha,\beta}(\eta_k)\|^2_{\mathrm{F}} \leq 3\kappa_k^2\Delta/4.
\end{equation}
Combining \eqref{theorem1_eq1} and \eqref{theorem1_eq2}, and using Assumption~\ref{ass-SNR} along with fact that $ C_{\mathrm{SNR}}$ is a sufficiently large constant, we obtain
\begin{align*}
\big\langle \widetilde{\mathbf{A}}^{\alpha,\beta}(\eta_k), \widetilde{\mathbf{B}}^{\alpha,\beta}(\eta_k) \big\rangle \geq  \kappa_k^2\Delta /16 - \kappa_k^2\Delta /64 - \kappa_k^2\Delta /64 \geq   \kappa_k^2\Delta /32.
\end{align*}
By the definition of $\mathcal{I}^*$,  it follows that 
\begin{equation}\label{eq-optimal}
a_{\mathcal{I}^*} = \big\langle\tilde{A}^{\alpha^*, \beta^*}(b_{\mathcal{I}^*}), \tilde{B}^{\alpha^*, \beta^*}(b_{\mathcal{I}^*}) \big\rangle \geq (\kappa_{\max}^{s,e})^2 \Delta/32,
\end{equation}
where
\[
\kappa_{\max}^{s,e} = \max \big\{\kappa_k \colon \min\{\eta_k - s, e - \eta_k\} \geq 3\Delta/4 \big\}.
\]
Therefore, if the threshold $\tau$ satisfies
\[
\tau <  \kappa^2 \Delta /32,
\]
Algorithm~\ref{offline-algorithm} will consistently detect the existence of any previously undetected change points within the interval.

\medskip \noindent
\textbf{Step 3.}  Suppose there exists at least one undetected change point $\eta_k \in (s,e)$ such that
\[
\min\{\eta_k - s, \eta_k - e\} \geq 3\Delta/4.
\]
Let $a_{\mathcal{I}}, b_{\mathcal{I}}$ and $\mathcal{I}^*$ be defined according to the procedure $\mathrm{SBS} \big((s, e), \tau, \mathcal{J} \big)$, and  denote  $\mathcal{I}^* =  (\alpha^{* \prime}, \beta^{* \prime}]$. To complete the induction step, it suffices to establish that there exists an undetected change point $\eta_k \in  (\alpha^{* \prime}, \beta^{* \prime}) $ satisfying
\begin{equation}\label{eq-step3.1}
\min\{\eta_k - \alpha^{* \prime}, \beta^{* \prime} - \eta _k \} \geq 3\Delta/4,
\end{equation}
and that
\begin{equation}\label{eq-step3.2}
|b_{\mathcal{I}^*} - \eta_k| \leq \tilde{\epsilon}.
\end{equation}

\medskip \noindent
\textbf{Step 3.1. Proof of \eqref{eq-step3.1}.}
Let
\begin{equation}\label{eq-interval-refined}
(\alpha^*, \beta^*] = (\alpha^{* \prime}+ (\beta^{*\prime} - \alpha^{*\prime})/64, \beta^{*\prime}-  (\beta^{* \prime} - \alpha^{* \prime})/64].
\end{equation}
Assume for contradiction that
\begin{equation}\label{theorem1_eq3}
\max_{\alpha^* < t < \beta^*} \big\|\widetilde{\mathbf{P}}^{\alpha^*, \beta^*}(t) \big\|_\mathrm{F}^2 < (\kappa_{\max}^{s,e})^2 \Delta/64.
\end{equation}
Then on the event $\mathcal{A}$, we obtain
\begin{align} 
& \max_{\alpha^* < t < \beta^*} \big\langle \widetilde{A}^{\alpha^*, \beta^*}(t), \widetilde{B}^{\alpha^*, \beta^*}(t) \big \rangle \nonumber\\
\leq &   \max_{\alpha^* < t < \beta^*} \big\|\widetilde{\mathbf{P}}^{\alpha^*, \beta^*}(t) \big\|_\mathrm{F}^2  + C_{\mathcal{A}} \log(T) \max_{\alpha^* < t < \beta^*} \big\|\widetilde{\mathbf{P}}^{\alpha^*, \beta^*}(t) \big\|_\mathrm{F}  + C_{\mathcal{A}} \log^{3/2}(T) n \sqrt{L} \nonumber\\
<  & (\kappa_{\max}^{s,e})^2 \Delta/64 +  C_{\mathcal{A}} \log(T)   \kappa_{\max}^{s,e} \sqrt{\Delta}/8 +  C_{\mathcal{A}} \log^{3/2}(T) n \sqrt{L}  \nonumber\\
\leq & (\kappa_{\max}^{s,e})^2 \Delta/64  +  (\kappa_{\max}^{s,e})^2 \Delta/128  +  (\kappa_{\max}^{s,e})^2 \Delta/128 = (\kappa_{\max}^{s,e})^2 \Delta/32,\nonumber
\end{align}
where the second inequality follows from \eqref{theorem1_eq3}, the third inequality follows from Assumption~\ref{ass-SNR} and the fact that $C_{\mathrm{SNR}}$ is a large enough constant. This contradicts the inequality \eqref{eq-optimal}.
Thus,
\begin{equation}\label{eq-step3.1-main}
\max_{\alpha^* < t < \beta^*} \big\|\widetilde{\mathbf{P}}^{\alpha^*, \beta^*}(t) \big\|_\mathrm{F}^2 \geq (\kappa_{\max}^{s,e})^2 \Delta/64.
\end{equation}

Now, observe
$(i)$ If $(\alpha^*, \beta^*)$ contains at least two change points, then $\beta^* - \alpha^* \geq \Delta$.
$(ii)$ If it contains exactly one change point $\eta_k$, but 
    $\min\{\eta_k - \alpha^*, \beta^* - \eta_k\} < \Delta/64$, then by Lemma~\ref{lemma-s-15}, we would have
\begin{align*}
\max_{\alpha^* < t < \beta^*} \big\|\widetilde{P}^{\alpha^*, \beta^*}(t) \big\|_{\mathrm{F}}^2
= & \big\|\widetilde{P}^{\alpha^*, \beta^*}(\eta_k) \big\|_{\mathrm{F}}^2 = \frac{(\eta_k - \alpha^*)(\beta^* - \eta_k)}{\beta^* - \alpha^*} \kappa_k^2  \nonumber\\
\leq  & \min\{\beta^* - \eta_k, \eta_k - \alpha^*\} \kappa_k^2 < (\kappa_{\max}^{s,e})^2 \Delta/64,
\end{align*} 
contradicting \eqref{eq-step3.1-main}. Therefore, it has to be the case that $\min\{\eta_k - \alpha^*, \beta^* - \eta_k\} \geq \Delta/64$. 
Moreover,  by \eqref{eq-interval-refined}, it holds that \[
\beta^{*, \prime} - \alpha^{*, \prime}  \geq \beta^* - \alpha^* \geq \Delta/64.
\] 
Then, using Assumption~\ref{ass-SNR}  and the the fact that $C_{\mathrm{SNR}}$ is large enough, we have that
\[
\tilde{\epsilon} \leq C_{\tilde{\epsilon}} \Delta \bigg( \frac{1}{C_{\mathrm{SNR}}^2 \log^{1/2 + 2\xi}(T)} + \frac{1}{C_{\mathrm{SNR}} \sqrt{nL^{1/2}} \log^{1+\xi}(T)} \bigg) \leq (\beta^{*, \prime} - \alpha^{*, \prime} )/64,
\]
Hence, by a similar argument as in \textbf{Step 1},  no previously detected change point lies in $(\alpha^*, \beta^*)$.  Note that by \eqref{eq-optimal}, there is at least one undetected change point in $(\alpha^*, \beta^*)$.

\medskip \noindent
\textbf{Step 3.2. Proof of  \eqref{eq-step3.2}.}  To this end, we apply  Lemma S.5 in \cite{wang2021optimal}. Define
\begin{equation} \label{eq-lemma-1}
\lambda = \max_{\alpha^* < t < \beta^*} \Big\vert \big\langle \widetilde{\mathbf{A}}^{\alpha^*, \beta^*}(t), \widetilde{\mathbf{B}}^{\alpha^*, \beta^*}(t) \big\rangle - \big\|\widetilde{\mathbf{P}}^{\alpha^*, \beta^*}(t)\big\|_{\mathrm{F}}^2  \Big\vert. 
\end{equation}
From \eqref{eq-step3.1-main}, Assumption~\ref{ass-SNR} and the fact that $C_{\mathrm{SNR}}$ is a sufficiently large constant, it follows that 
\begin{align}\label{eq-step3.2-main} 
   c_2 \max_{\alpha^* < t < \beta^*}  \big\|  \widetilde{\mathbf{P}}^{\alpha^*, \beta^*} (t) \big \|_{\mathrm{F}}^2 /2  \geq  &\max\big\{C_{\mathcal{A}} \log(T) \max_{\alpha^* < t < \beta^*}   \big\|  \widetilde{\mathbf{P}}^{\alpha^*,
   \beta^*} (t) \big \|_{\mathrm{F}},\nonumber \\
    & \hspace{1cm} 
    C_{\mathcal{A}} \log^{3/2}(T) n\sqrt{L} \Big\},
\end{align}
where $c_2>0$ is a small enough constant. By the definition of the event $\mathcal{A}$, we obtain
\begin{align}\label{eq-lemma-2}
    \lambda \leq &  C_{\mathcal{A}} \log(T) \big\{ \max_{ \alpha^* < t < \beta^*} \big\|  \widetilde{\mathbf{P}}^{\alpha^*, \beta^*} (t) \big \|_{\mathrm{F}} + \log^{1/2}(T) n\sqrt{L} \big) \big\}  \leq    c_2 \max_{ \alpha^* < t < \beta^*} \big\|  \widetilde{\mathbf{P}}^{\alpha^*, \beta^*} (t) \big \|_{\mathrm{F}}^2,
\end{align}
where the last inequality follows from \eqref{eq-step3.2-main}.
Note that \eqref{eq-optimal}, \eqref{eq-lemma-1} and \eqref{eq-lemma-2} verify conditions (2), (3),  (4) of  Lemma S.5 in \cite{wang2021optimal}, respectively. Therefore,  Lemma S.5 in \cite{wang2021optimal} implies that there exists an undetected change point $\eta_k$ within $(s, e)$ such that
\begin{align}\label{eq-lemma-3}
\vert \eta_k - b_{\mathcal{I}^*}\vert  \leq \frac{C_3 \Delta \lambda}{\|\widetilde{P}^{\alpha^*, \beta^*}(\eta_k)\|^2_{\mathrm{F}}} \quad \mbox{and} \quad \|\widetilde{P}^{\alpha^*, \beta^*}(\eta_k)\|^2_{\mathrm{F}} \geq c_4 \max_{\alpha^* \leq t \leq \beta^*} \|\widetilde{P}^{\alpha^*, \beta^*}(t)\|^2_{\mathrm{F}},
\end{align}
where $C_3, c_4 > 0$ are constants. Then combining \eqref{eq-lemma-2} and \eqref{eq-lemma-3}, we can derive that
\begin{align}
\vert \eta_k - b_{\mathcal{I}^*}\vert    \leq & \frac{C_3 C_{\mathcal{A}}\Delta  \log(T) \big\{ \max_{ \alpha^* < t < \beta^*} \big\|  \widetilde{\mathbf{P}}^{\alpha^*, \beta^*} (t) \big \|_{\mathrm{F}} + \log^{1/2}(T) n\sqrt{L} \big) \big\}\ }{ c_4 \max_{\alpha^* \leq t \leq \beta^*} \|\widetilde{\mathbf{P}}^{\alpha^*, \beta^*}(t)\|^2_{\mathrm{F}}} \nonumber\\
= &  \frac{C_3 C_{\mathcal{A}}  \log(T) }{c_4} \bigg\{ \frac{\Delta}{ \max_{ \alpha^* < t < \beta^*} \big\|  \widetilde{\mathbf{P}}^{\alpha^*, \beta^*} (t) \big \|_{\mathrm{F}}} + \frac{\Delta \log^{1/2}(T) n\sqrt{L}  }{ \max_{\alpha^* \leq t \leq \beta^*} \|\widetilde{P}^{\alpha^*, \beta^*}(t)\|^2_{\mathrm{F}}}\bigg\} \nonumber\\
\leq &  \frac{C_3 C_{\mathcal{A}}  \log(T) }{c_4} \bigg\{ \frac{8\sqrt{\Delta} }{ \kappa } + \frac{64\log^{1/2}(T) n\sqrt{L} }{ \kappa^2 } \bigg\} \nonumber\\
\leq & C_5  \log(T) \bigg\{ \frac{\sqrt{\Delta} }{ \kappa } + \frac{\log^{1/2}(T) n\sqrt{L} }{ \kappa^2 } \bigg\}, \nonumber
\end{align}
where the second inequality follows form \eqref{eq-step3.1-main} and $C_5 > 0$ is an constant. This completes the induction step and therefore, the proof.

\end{proof}

\begin{lemma}\label{lemma-s-2}
Let $\{\mathbf{A}(t)\}_{t \in [T]}$ follow $\mathrm{D}\mbox{-}\mathrm{MRDPGs}$ as in Definition~\ref{def-umrdpg-f-dynamic}. Let $ \mathbf{V} \in \mathbb{R}^{n\times n \times L}$ and $ \{w_t\}_{t=1}^T \subset \mathbb{R} $ satisfy $\sum_{t=1}^T w_t^2 = 1$. Then for any $ \varepsilon > 0 $, it holds that
\[
\mathbb{P} \bigg( \bigg \vert   \sum_{t=1}^T w_t \big\langle \mathbf{A}(t) - \mathbf{P}(t), \mathbf{V}/ \| \mathbf{V}\|_{\mathrm{F}} \big\rangle \bigg \vert  \geq \varepsilon \bigg) 
\leq 2 \exp \bigg( - c \frac{\varepsilon^2}{1  + \varepsilon  \| \mathbf{V}\|_{\mathrm{F}}^{-1}\| \mathbf{V}\|_{\infty}\max_{1 \leq t \leq T} \vert w_t \vert }  \bigg),
\]
where $c> 0$ is an absolute constant and  $ \| \mathbf{V}\|_{\infty}  = \max_{i, j \in [n], l \in [L]} \vert V_{i, j, l}\vert $.
\end{lemma}

\begin{proof}
By definition of the tensor inner product, we have that
\[
\sum_{t=1}^T w_t \big\langle \mathbf{A}(t) - \mathbf{P}(t), \mathbf{V}/ \| \mathbf{V}\|_{\mathrm{F}} \big\rangle  
=  \sum_{t=1}^{T} \sum_{i = 1}^n \sum_{j=1}^n \sum_{l=1}^L   \| \mathbf{V}\|_{\mathrm{F}}^{-1} w_t \mathbf{V}_{i, j, l}  \Big\{ \big(\mathbf{A}(t) \big)_{i, j, l} - \big( \mathbf{P}(t) \big)_{i, j, l}  \Big\}. 
\]
We can derive  that
\[
 \sum_{t=1}^{T} \sum_{i = 1}^n \sum_{j=1}^n \sum_{l=1}^L   \| \mathbf{V}\|_{\mathrm{F}}^{-2} w_t^2 \mathbf{V}_{i, j, l}^2 = 1,
\]
and 
\[
\max_{t \in [T], i, j \in [n], l \in [L]} \Big\{  \| \mathbf{V}\|_{\mathrm{F}}^{-1} w_t \mathbf{V}_{i, j, l}\Big\} \leq  \| \mathbf{V}\|_{\mathrm{F}}^{-1}  \| \mathbf{V}\|_{\infty}  \max_{1 \leq t \leq T}\vert w_t \vert.
\]
Since 
\[
\{ \big(\mathbf{A}(t)\big)_{i, j, l} - \big(\mathbf{P}(t)\big)_{i, j, l} \}_{i,j\in [n], l \in [L], t \in [T]}
\]
are mutually independent centered Bernoulli random variables By Bernstein inequality \citep[e.g.~Theorem 2.8.2 in][]{vershynin2018high}, it holds with an absolute constant $c_0 > 0$ that 
\[
\mathbb{P} \bigg( \bigg \vert   \sum_{t=1}^T w_t \big\langle \mathbf{A}(t) - \mathbf{P}(t), \mathbf{V}/ \| \mathbf{V}\|_{\mathrm{F}} \big\rangle \bigg \vert  \geq \varepsilon \bigg) 
\leq 2 \exp \bigg( - c_0 \frac{\varepsilon^2}{1  + \varepsilon  \| \mathbf{V}\|_{\mathrm{F}}^{-1} \| \mathbf{V}\|_{\infty}\max_{1 \leq t \leq T} \vert w_t \vert } \bigg),
\]
which completes the proof. 

\end{proof}

\begin{lemma}\label{lemma-s-15}

Suppose the adjacency tensor sequence $ \{\mathbf{B}(t)\}_{t \in [T]} \subset \{0, 1\}^{n \times n \times L}$ is generated according to Definition~\ref{def-umrdpg-f-dynamic} and satisfy Model~\ref{model-1}. For any $ 0 \leq s < t< e \leq T$, let $\widetilde{\mathbf{P}}^{s, e }(t)$ be defined 
as in \eqref{def-average}.
If $(s, e)$ contains exactly one change point $\eta_k$, then for any $t \in (s, e)$
\[
\big\|\widetilde{\mathbf{P}}^{s,e}(t)\big\|_{\mathrm{F}}^2 
=
\begin{cases}
\frac{t - s}{(e - s)(e - t)}(e- \eta_k)^2 
\kappa_k^2,
& s <t \le \eta_k, \\
\frac{e - t}{(e - s)(t - s)}(\eta_k - s)^2 
\kappa_k^2,
& \eta_k <  t <e.
\end{cases}
\]
\end{lemma}

\begin{proof}
This follows directly the definition of $\widetilde{\mathbf{P}}^{s,e}(t)$ in \eqref{def-average}. 
\end{proof}

\section{Proofs of Theorems~\ref{theorem-inference} and \ref{theorem-inference-cont}}\label{proof-theorem-inference}
\begin{proof}

\medskip
\noindent
\textbf{Step 1.} Preliminary bounds.
We first define the event
\[
\mathcal{A}   = \Big\{ \widetilde{K} = K \mbox{ and } \vert \widetilde{\eta}_k - \eta_k \vert  \leq \epsilon_k, \, \forall k \in [K] \Big\}, \quad \mbox{where} \quad  \epsilon_k = C_{\epsilon}    \frac{\log(T)}{\kappa_k^2}.
\]
Then by Theorem~\ref{theorem-2}, it holds that  
\[
\mathbb{P}\{ \mathcal{A}\} \geq 1 - C_0T^{-c_0},
\]
and $C_{\tilde{\epsilon}}, C_0, c_0 >0$ are absolute constants. Since $\mathcal{A}$ holds with probability tending to $1$ as $T\to \infty$,we
condition the remainder of the proof on  $\mathcal{A}$.

From $\mathcal{A}$, Assumption~\ref{ass-SNR} and the fact that $C_{\mathrm{SNR}}$ is a sufficiently large constant, we have for all $k \in [K]$ that
$ \eta_k \in [\widetilde{\eta}_{k-1}, \widetilde{\eta}_{k+1}]$, 
\begin{equation}\label{theorem-3-1.1}
    \eta_k - \widetilde{\eta}_{k-1} \geq
  \eta_k - \eta_{k-1} - \vert \eta_{k-1} - \widetilde{\eta}_{k-1}\vert 
   \geq  \Delta - C_{\epsilon} \frac{\log(T)}{\kappa^2}  \geq  \Delta - \Delta/6  = 5\Delta/6,
\end{equation}
and 
\begin{equation}\label{theorem-3-1.2}
\widetilde{\eta}_{k+1} - \eta_k \geq  \eta_{k+1} - \eta_k - \vert \eta_{k+1} - \widetilde{\eta}_{k+1} \vert \geq  \Delta - C_{\epsilon} \frac{\log(T)}{\kappa^2}  \geq  \Delta - \Delta/6  = 5\Delta/6.
\end{equation}
Similarly, we can derive that for any $k \in [K+1]$
\begin{equation}\label{theorem-3-1.3}
\widetilde{\eta}_k - \widetilde{\eta}_{k-1} \geq  2\Delta/3.
\end{equation}
As a result, the working interval 
\[ 
(\tilde{s}_k,\tilde{e}_k] =  (\widetilde{\eta}_{k-1}  +(\widetilde{\eta}_k - \widetilde{\eta}_{k-1})/2,\;
   \widetilde{\eta}_{k+1} - (\widetilde{\eta}_{k+1} - \widetilde{\eta}_k)/2],
\]
contains exactly one change point $\eta_k$.

Next, by Theorem 4 and Lemma 5 in \cite{wang2023multilayer}, and Assumption~\ref{ass_X_f_Q_f} $(i)$ and $(iii)$, for any $k \in [K+1]$, we have
\[
\big\| \widehat{\mathbf{P}}^{\widetilde{\eta}_{k-1},\widetilde{\eta}_{k}} - \mathbf{P}^{\widetilde{\eta}_{k-1},\widetilde{\eta}_{k}} \big\|_{\mathrm{F}} = O_p \bigg( \sqrt{\frac{(d^2m_{\widetilde{\eta}_{k-1}, \widetilde{\eta}_{k}} + nd + Lm_{\widetilde{\eta}_{k-1}, \widetilde{\eta}_{k}} )   \log(T)}{\widetilde{\eta}_{k} - \widetilde{\eta}_{k-1}}} \bigg), 
\]
for some absolute constant  $C_{1}>0$. For any $k \in [K+1]$,  by \eqref{theorem-3-1.1} and \eqref{theorem-3-1.2}, each interval $(\widetilde{\eta}_{k-1},\widetilde{\eta}_{k})$
contains at most two true change points $\eta_k-1$ and $\eta_k$. Consequently, we have that  
\[
    m_{\widetilde{\eta}_{k-1},\widetilde{\eta}_{k}} \leq m_{k-1} + m_k + m_{k+1}\leq 3 m_{\max}.
\]
Thus, it holds that
\begin{equation}\label{def-event-A-2}
\big\| \widehat{\mathbf{P}}^{\widetilde{\eta}_{k-1},\widetilde{\eta}_{k}} - \mathbf{P}^{\widetilde{\eta}_{k-1},\widetilde{\eta}_{k}} \big\|_{\mathrm{F}} \leq C_{2} \sqrt{\frac{ (d^2m_{\max} + nd+Lm_{\max})   \log(T)}{\widetilde{\eta}_k -\widetilde{\eta}_{k-1}}}, \quad  \forall k \in [K+1],
\end{equation}
with an absolute constant $C_{2}>0$.

\medskip
\noindent
\textbf{Step 2.} Characterization of bias.
From \eqref{theorem-3-1.1} and \eqref{theorem-3-1.2}, for any $k \in [K]$, the interval $ (\widetilde{\eta}_{k-1}, \widetilde{\eta}_{k+1})$ may contain one, two or three change points.
One example that contains three change points is illustrated in the figure below. We analyze the biases in this case. The analyses for the other scenarios are similar but simpler and therefore omitted.

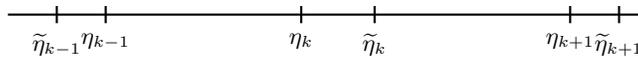
\vspace{0.3cm}
\begin{center}
\begin{tikzpicture}[scale=6.5,decoration=brace]
\draw[-, thick] (-0.65,0) -- (0.65,0); 
\foreach \x/\xtext in {-0.55/$\widetilde{\eta}_{k-1}$,-0.45/$\eta_{k-1}$,-0.05/$\eta_{k}$,0.1/$\widetilde{\eta}_{k}$,0.5/$\eta_{k+1}$,0.6/$\widetilde{\eta}_{k+1}$}
\draw[thick] (\x,0.5pt) -- (\x,-0.5pt) node[font = {\footnotesize}, below] {\xtext};
\end{tikzpicture}
\end{center}

In the following, we analyze three types of bias terms.  Denote $\alpha_T =\log(T)$, then $\alpha_T \to \infty$ as $T \to \infty$. Observe that 
\begin{align}\label{eq-inference-1}
 &   \big\| \mathbf{P}(\eta_{k})  -  \mathbf{P}^{\widetilde{\eta}_{k-1},\widetilde{\eta}_{k}}  \big\|_{\mathrm{F}} \nonumber\\
=  &  \bigg\| \mathbf{P}(\eta_{k})  -  
\frac{ \eta_{k-1} - \widetilde{\eta}_{k-1} }{\widetilde{\eta}_k - \widetilde{\eta}_{k-1}   }
  \mathbf{P}(\eta_{k-1})  - \frac{ \eta_{k} - \eta_{k-1}}{\widetilde{\eta}_k - \widetilde{\eta}_{k-1} }
  \mathbf{P}(\eta_{k}) -  \frac{ \widetilde{\eta}_{k} - \eta_{k}}{\widetilde{\eta}_k - \widetilde{\eta}_{k-1} }
  \mathbf{P}(\eta_{k+1})    \bigg \|_{\mathrm{F}} \nonumber\\
 \leq  & \frac{\eta_{k-1} - \widetilde{\eta}_{k-1} }{\widetilde{\eta}_k - \widetilde{\eta}_{k-1}   } \big\|   \mathbf{P}(\eta_{k})   - 
  \mathbf{P}(\eta_{k-1}) \big \|_{\mathrm{F}} +   \frac{ \widetilde{\eta}_{k} - \eta_{k}}{\widetilde{\eta}_k - \widetilde{\eta}_{k-1} }
   \big\|  \mathbf{P}(\eta_{k+1})  -  \mathbf{P}(\eta_{k})   \big \|_{\mathrm{F}}  \nonumber\\
   \leq &  \frac{3  C_{\epsilon} \log(T)}{2  \Delta \kappa_{k-1}^2}  \kappa_{k-1}+ \frac{ 3 C_{\epsilon} \log(T)}{2  \Delta \kappa_{k}^2}  \kappa_{k}
   \leq \alpha_T^{-1}\kappa_{k},
\end{align}
where the second inequity follows from the event $\mathcal{A}$ and \eqref{theorem-3-1.3}, and the last inequality follows from Assumption~\ref{ass-SNR} and the fact that $C_{\mathrm{SNR}}$ is a large enough constant. 
Similarly, we have that 
\begin{align}\label{eq-inference-2}
   \big\| \mathbf{P}(\eta_{k+1})  -  \mathbf{P}^{\widetilde{\eta}_{k},\widetilde{\eta}_{k+1}}  \big\|_{\mathrm{F}}
=  &  \bigg\| \mathbf{P}(\eta_{k+1})     - \frac{ \eta_{k+1} -\widetilde{\eta}_{k}  }{\widetilde{\eta}_{k+1} - \widetilde{\eta}_{k} }
  \mathbf{P}(\eta_{k+1}) -  \frac{ \widetilde{\eta}_{k+1} - \eta_{k+1}}{\widetilde{\eta}_{k+1} - \widetilde{\eta}_{k} }
  \mathbf{P}(\eta_{k+2})    \bigg \|_{\mathrm{F}} \nonumber\\
 =  & \frac{ \widetilde{\eta}_{k+1} - \eta_{k+1}}{\widetilde{\eta}_{k+1} - \widetilde{\eta}_{k} } \big\| \mathbf{P}(\eta_{k+1}) - 
  \mathbf{P}(\eta_{k+2})    \big\|_{\mathrm{F}} \nonumber\\
   \leq & \frac{3  C_{\epsilon} \log(T)}{2\Delta\kappa_{k+1}^2} \kappa_{k+1}  \leq \alpha_T^{-1}\kappa_{k},
\end{align}
and
\begin{align}\label{eq-inference-3}
  & \big\| \mathbf{P}(\eta_{k+1})  -  \mathbf{P}^{\widetilde{\eta}_{k-1},\widetilde{\eta}_{k}}  \big\|_{\mathrm{F}} \nonumber \\
=  &  \bigg\| \mathbf{P}(\eta_{k+1})  -  
\frac{ \eta_{k-1} - \widetilde{\eta}_{k-1} }{\widetilde{\eta}_k - \widetilde{\eta}_{k-1}   }
  \mathbf{P}(\eta_{k-1})  - \frac{ \eta_{k} - \eta_{k-1}}{\widetilde{\eta}_k - \widetilde{\eta}_{k-1} }
  \mathbf{P}(\eta_{k}) -  \frac{ \widetilde{\eta}_{k} - \eta_{k}}{\widetilde{\eta}_k - \widetilde{\eta}_{k-1} }
  \mathbf{P}(\eta_{k+1})    \bigg \|_{\mathrm{F}} \nonumber\\
 \leq  & \frac{\eta_{k-1} - \widetilde{\eta}_{k-1} }{\widetilde{\eta}_k - \widetilde{\eta}_{k-1}   } \big\|   \mathbf{P}(\eta_{k+1})   - 
  \mathbf{P}(\eta_{k-1}) \big \|_{\mathrm{F}} +   \frac{ \eta_{k} - \eta_{k-1}}{\widetilde{\eta}_k - \widetilde{\eta}_{k-1} }
   \big\|  \mathbf{P}(\eta_{k+1})  -  \mathbf{P}(\eta_{k})   \big \|_{\mathrm{F}}  \nonumber\\
   \leq & \frac{ 3  C_{\epsilon} \log(T)}{2\Delta\kappa_{k-1}^2}   (\kappa_{k-1} + \kappa_k)+  \kappa_{k}
   \leq \alpha_T^{-1}\kappa_{k} +\kappa_k \leq C_3 \kappa_k,
\end{align}
for some constant $C_3 >0$.

\medskip
\noindent
\textbf{Step 3.} Uniform tightness of $\kappa_k^2 |\widehat{\eta}_k - \eta_k|$. In this step, we show that  $\kappa_k^2 |\widehat{\eta}_k - \eta_k|  = O_p (1)$. Let $ r = \widehat{\eta}_k - \eta_k $ and without loss of generality, assume $ r \geq 0 $. Our goal is to establish that 
\[
r \kappa_k^2  = O_p (1)
\]
If $r\kappa_k^2  < 1 $, the conclusion holds trivially.  Thus, for the remainder of the argument, we assume that~$r\kappa_k^2  \geq 1 $.
Since $ \widehat{\eta}_k = \eta_k + r $,  it follows that
\[
\mathcal{Q}_k(\eta_k+ r)  -  \mathcal{Q}_k(\eta_k)   \leq 0.
\]
Now observe that 
\begin{align}
\mathcal{Q}_k(\eta_k+ r)  -  \mathcal{Q}_k(\eta_k)   = &  \sum_{t=\eta_k+1}^{\eta_k + r} \big\| \mathbf{A}(t)  - \widehat{\mathbf{P}}^{\widetilde{\eta}_{k-1}, \widetilde{\eta}_{k}} \big\|_{\mathrm{F}}^2 - \big\| \mathbf{A}(t)  - \widehat{\mathbf{P}}^{\widetilde{\eta}_{k}, \widetilde{\eta}_{k+1}} \big\|_{\mathrm{F}}^2 \nonumber\\
= &  \sum_{t=\eta_k+1}^{\eta_k + r}  \Big\{ \big\| \mathbf{A}(t)  - \widehat{\mathbf{P}}^{\widetilde{\eta}_{k-1}, \widetilde{\eta}_{k}} \big\|_{\mathrm{F}}^2 - \big\| \mathbf{A}(t)  - \mathbf{P}^{\widetilde{\eta}_{k-1}, \widetilde{\eta}_{k}} \big\|_{\mathrm{F}}^2 \Big\} \nonumber\\
& \hspace{0.5cm} - \sum_{t=\eta_k+1}^{\eta_k + r}   \Big\{ \big\| \mathbf{A}(t)  - \widehat{\mathbf{P}}^{\widetilde{\eta}_{k}, \widetilde{\eta}_{k+1}} \big\|_{\mathrm{F}}^2 - \big\| \mathbf{A}(t)   - \mathbf{P}^{\widetilde{\eta}_{k}, \widetilde{\eta}_{k+1}} \big\|_{\mathrm{F}}^2  \Big\} \nonumber\\
&  \hspace{0.5cm}  + \sum_{t=\eta_k+1}^{\eta_k + r}    \Big\{  \big\| \mathbf{A}(t)   - \mathbf{P}^{\widetilde{\eta}_{k-1}, \widetilde{\eta}_{k}} \big\|_{\mathrm{F}}^2  -  \big\| \mathbf{A}(t) - \mathbf{P}(\eta_k) \big\|_{\mathrm{F}}^2  \Big\} \nonumber\\
&  \hspace{0.5cm}  - \sum_{t=\eta_k+1}^{\eta_k + r}     \Big\{  \big\| \mathbf{A}(t)  - \mathbf{P}^{\widetilde{\eta}_{k}, \widetilde{\eta}_{k+1}} \big\|_{\mathrm{F}}^2  -  \big\| \mathbf{A}(t) - \mathbf{P}(\eta_{k+1}) \big\|_{\mathrm{F}}^2  \Big\} \nonumber\\
&   \hspace{0.5cm} + \sum_{t=\eta_k+1}^{\eta_k + r}      \Big\{  \big\| \mathbf{A}(t)  - \mathbf{P}(\eta_{k})\big\|_{\mathrm{F}}^2  -  \big\| \mathbf{A}(t) - \mathbf{P}(\eta_{k+1}) \big\|_{\mathrm{F}}^2  \Big\} \nonumber\\
= & I - II + III - IV + V. \nonumber
\end{align}
Therefore, we have that
\begin{equation}\label{eq-theorem-3-main}
V \leq -I + II - III + IV \leq |I| + |II| + |III| + |IV|. 
\end{equation}

\medskip
\noindent
\textbf{Step 3.1.} Order of magnitude of $I$. We start by analyzing the term
\begin{align}\label{eq-inference-I.00} 
I = &  \sum_{t=\eta_k+1}^{\eta_k + r} \Big\{ \big\| \mathbf{A}(t)  - \widehat{\mathbf{P}}^{\widetilde{\eta}_{k-1}, \widetilde{\eta}_{k}} \big\|_{\mathrm{F}}^2 - \big\| \mathbf{A}(t)  - \mathbf{P}^{\widetilde{\eta}_{k-1}, \widetilde{\eta}_{k}} \big\|_{\mathrm{F}}^2 \Big\}
\nonumber\\
= & \sum_{t=\eta_k+1}^{\eta_k + r} \big\|   \widehat{\mathbf{P}}^{\widetilde{\eta}_{k-1}, \widetilde{\eta}_{k}} -  \mathbf{P}^{\widetilde{\eta}_{k-1}, \widetilde{\eta}_{k}} \big\|_{\mathrm{F}}^2 -  2 \sum_{t=\eta_k+1}^{\eta_k + r}  \big\langle \mathbf{A}(t)  -   \mathbf{P}^{\widetilde{\eta}_{k-1}, \widetilde{\eta}_{k}}, \widehat{\mathbf{P}}^{\widetilde{\eta}_{k-1}, \widetilde{\eta}_{k}} -  \mathbf{P}^{\widetilde{\eta}_{k-1}, \widetilde{\eta}_{k}} \big\rangle \nonumber\\
 = & \sum_{t=\eta_k+1}^{\eta_k + r} \big\|   \widehat{\mathbf{P}}^{\widetilde{\eta}_{k-1}, \widetilde{\eta}_{k}} -  \mathbf{P}^{\widetilde{\eta}_{k-1}, \widetilde{\eta}_{k}} \big\|_{\mathrm{F}}^2 -  2 \sum_{t=\eta_k+1}^{\eta_k + r}  \big\langle \mathbf{A}(t)  -   \mathbf{P}(\eta_{k+1}), \widehat{\mathbf{P}}^{\widetilde{\eta}_{k-1}, \widetilde{\eta}_{k}} -  \mathbf{P}^{\widetilde{\eta}_{k-1}, \widetilde{\eta}_{k}} \big\rangle  \nonumber\\
 & \hspace{0.5cm} -  2 \sum_{t=\eta_k+1}^{\eta_k + r}  \big\langle \mathbf{P}(\eta_{k+1})  -   \mathbf{P}^{\widetilde{\eta}_{k-1}, \widetilde{\eta}_{k}}, \widehat{\mathbf{P}}^{\widetilde{\eta}_{k-1}, \widetilde{\eta}_{k}} -  \mathbf{P}^{\widetilde{\eta}_{k-1}, \widetilde{\eta}_{k}} \big\rangle \nonumber\\
= & I.1 - 2I.2 - 2I.3 . 
\end{align}
By \eqref{def-event-A-2}, we have that 
\begin{align}\label{eq-inference-I.0}
\big\|   \widehat{\mathbf{P}}^{\widetilde{\eta}_{k-1}, \widetilde{\eta}_{k}} -  \mathbf{P}^{\widetilde{\eta}_{k-1}, \widetilde{\eta}_{k}} \big\|_{\mathrm{F}}^2  \leq & C_{2}^2 \frac{ (d^2m_{\max} + nd+Lm_{\max})   \log(T)}{\widetilde{\eta}_k -\widetilde{\eta}_{k-1}} \nonumber\\
\leq &   3C_{2}^2 \frac{ (d^2m_{\max} + nd+Lm_{\max})   \log(T)}{2\Delta} \leq 
 \alpha_{T}^{-1}\kappa_{k}^2,
\end{align}
where the second inequality is by \eqref{theorem-3-1.3} and the last inequality follows from Assumption~\ref{ass-SNR} and the fact that $C_{\mathrm{SNR}}$ is a sufficiently large constant. 
This yields that  
\begin{align}\label{eq-inference-I.1}
|I.1| = &  \sum_{t=\eta_k+1}^{\eta_k + r} \big\|   \widehat{\mathbf{P}}^{\widetilde{\eta}_{k-1}, \widetilde{\eta}_{k}} -  \mathbf{P}^{\widetilde{\eta}_{k-1}, \widetilde{\eta}_{k}} \big\|_{\mathrm{F}}^2  = O_p( r \alpha_{T}^{-1}\kappa_{k}^2),
\end{align}
We now turn to the term $I.2$ in \eqref{eq-inference-I.00}. By Lemma~\ref{lemma-s-2} and \eqref{eq-inference-I.0}, we obtain that 
\begin{align}\label{eq-inference-I.2}
\vert I.2 \vert  =   O_p \Big( r^{1/2} \big\|   \widehat{\mathbf{P}}^{\widetilde{\eta}_{k-1}, \widetilde{\eta}_{k}} -  \mathbf{P}^{\widetilde{\eta}_{k-1}, \widetilde{\eta}_{k}} \big\|_{\mathrm{F}} \Big) = O_p (r^{1/2} 
\alpha_T^{-1/2} \kappa_k ). 
\end{align}
Next, by the Cauchy--Schwarz inequality, we derive that
\begin{align}\label{eq-inference-I.3}
\vert I.3\vert  \leq  r \big\| \mathbf{P}(\eta_{k+1})  -   \mathbf{P}^{\widetilde{\eta}_{k-1}, \widetilde{\eta}_{k}}\big\|_{\mathrm{F}}  \big\|\widehat{\mathbf{P}}^{\widetilde{\eta}_{k-1}, \widetilde{\eta}_{k}} -  \mathbf{P}^{\widetilde{\eta}_{k-1}, \widetilde{\eta}_{k}} \big\|_{\mathrm{F}} =  O_p( r \alpha_T^{-1/2} \kappa_k ^2),
\end{align}
where the last inequality follows from  \eqref{eq-inference-3} and \eqref{eq-inference-I.0}.

Combining \eqref{eq-inference-I.00}, \eqref{eq-inference-I.1}, \eqref{eq-inference-I.2} and \eqref{eq-inference-I.3}, we conclude that
\begin{align}\label{eq-inference-I}
(I) =    o_p (r \kappa_k^2 + r^{1/2}\kappa_k).
\end{align}

\medskip
\noindent
\textbf{Step 3.2.} Order of magnitude of $ II$. 
We now analyze the term
\begin{align}\label{eq-inference-II.00}
II = &  \sum_{t=\eta_k+1}^{\eta_k + r} \Big\{  \big\| \mathbf{A}(t)  - \widehat{\mathbf{P}}^{\widetilde{\eta}_{k}, \widetilde{\eta}_{k+1}} \big\|_{\mathrm{F}}^2 - \big\| \mathbf{A}(t)   - \mathbf{P}^{\widetilde{\eta}_{k}, \widetilde{\eta}_{k+1}} \big\|_{\mathrm{F}}^2  \Big\}
\nonumber\\
= & \sum_{t=\eta_k+1}^{\eta_k + r} \big\|   \widehat{\mathbf{P}}^{\widetilde{\eta}_{k}, \widetilde{\eta}_{k+1}} -  \mathbf{P}^{\widetilde{\eta}_{k}, \widetilde{\eta}_{k+1}} \big\|_{\mathrm{F}}^2 -  2 \sum_{t=\eta_k+1}^{\eta_k + r}  \big\langle \mathbf{A}(t)  -   \mathbf{P}^{\widetilde{\eta}_{k}, \widetilde{\eta}_{k+1}}, \widehat{\mathbf{P}}^{\widetilde{\eta}_{k}, \widetilde{\eta}_{k+1}} -  \mathbf{P}^{\widetilde{\eta}_{k}, \widetilde{\eta}_{k+1}} \big\rangle \nonumber\\
 = & \sum_{t=\eta_k+1}^{\eta_k + r} \big\|   \widehat{\mathbf{P}}^{\widetilde{\eta}_{k}, \widetilde{\eta}_{k+1}} -  \mathbf{P}^{\widetilde{\eta}_{k}, \widetilde{\eta}_{k+1}} \big\|_{\mathrm{F}}^2 -  2 \sum_{t=\eta_k+1}^{\eta_k + r}  \big\langle \mathbf{A}(t)  -   \mathbf{P}(\eta_{k+1}), \widehat{\mathbf{P}}^{\widetilde{\eta}_{k}, \widetilde{\eta}_{k+1}} -  \mathbf{P}^{\widetilde{\eta}_{k}, \widetilde{\eta}_{k+1}} \big\rangle  \nonumber\\
 & \hspace{0.5cm} -  2 \sum_{t=\eta_k+1}^{\eta_k + r}  \big\langle \mathbf{P}(\eta_{k+1})  -   \mathbf{P}^{\widetilde{\eta}_{k}, \widetilde{\eta}_{k+1}}, \widehat{\mathbf{P}}^{\widetilde{\eta}_{k}, \widetilde{\eta}_{k+1}} -  \mathbf{P}^{\widetilde{\eta}_{k}, \widetilde{\eta}_{k+1}} \big\rangle \nonumber\\
= & II.1 - 2II.2 - 2II.3 . 
\end{align}
By \eqref{def-event-A-2}, we have that 
\begin{align}\label{eq-inference-II.0}
\big\|   \widehat{\mathbf{P}}^{\widetilde{\eta}_{k}, \widetilde{\eta}_{k+1}} -  \mathbf{P}^{\widetilde{\eta}_{k}, \widetilde{\eta}_{k+1}} \big\|_{\mathrm{F}}^2  \leq & C_{2}^2 \frac{ (d^2m_{\max} + nd+Lm_{\max})   \log(T)}{\widetilde{\eta}_{k+1} -\widetilde{\eta}_{k}} \nonumber\\
\leq &  3 C_{2}^2 \frac{ (d^2m_{\max} + nd+Lm_{\max})   \log(T)}{2\Delta} \leq 
 \alpha_{T}^{-1}\kappa_{k}^2,
\end{align}
where the second inequality follows from \eqref{theorem-3-1.3} and the last inequality follows from Assumption~\ref{ass-SNR} and the fact that $C_{\mathrm{SNR}}$ is a sufficiently large constant. 
It then follows that  
\begin{align}\label{eq-inference-II.1}
|II.1| = &  \sum_{t=\eta_k+1}^{\eta_k + r} \big\|   \widehat{\mathbf{P}}^{\widetilde{\eta}_{k}, \widetilde{\eta}_{k+1}} -  \mathbf{P}^{\widetilde{\eta}_{k}, \widetilde{\eta}_{k+1}} \big\|_{\mathrm{F}}^2  =  O_p( r \alpha_{T}^{-1}\kappa_{k}^2),
\end{align}
To control $II.2$, by Lemma~\ref{lemma-s-2} and \eqref{eq-inference-II.0},  we obtain that 
\begin{align}\label{eq-inference-II.2}
\vert II.2 \vert  = O_p\Big( r^{1/2}  \big\|   \widehat{\mathbf{P}}^{\widetilde{\eta}_{k}, \widetilde{\eta}_{k+1}} -  \mathbf{P}^{\widetilde{\eta}_{k}, \widetilde{\eta}_{k+1}} \big\|_{\mathrm{F}} \Big)  =   O_p( r^{1/2}  \alpha_T^{-1/2} \kappa_k). 
\end{align}
Next, by the Cauchy--Schwarz inequality, we derive that
\begin{align}\label{eq-inference-II.3}
\vert II.3 \vert  \leq  r \big\| \mathbf{P}(\eta_{k+1})  -   \mathbf{P}^{\widetilde{\eta}_{k}, \widetilde{\eta}_{k+1}}\big\|_{\mathrm{F}}  \big\|\widehat{\mathbf{P}}^{\widetilde{\eta}_{k-1}, \widetilde{\eta}_{k}} -  \mathbf{P}^{\widetilde{\eta}_{k}, \widetilde{\eta}_{k+1}} \big\|_{\mathrm{F}}  =  O_p( r \alpha_T^{-3/2} \kappa_k ^2),
\end{align}
where the last inequality follows from \eqref{eq-inference-2} and \eqref{eq-inference-II.0}.

Combining \eqref{eq-inference-II.00}, \eqref{eq-inference-II.1}, \eqref{eq-inference-II.2} and \eqref{eq-inference-II.3}, we conclude that
\begin{align}\label{eq-inference-II}
\vert II \vert  =   o_p (r \kappa_k^2 + r^{1/2}\kappa_k).
\end{align}

\medskip
\noindent
\textbf{Step 3.3.} Order of magnitude of $ III$. We now analyze the term
\begin{align}\label{eq-inference-III.0}
 III = & \sum_{t=\eta_k+1}^{\eta_k + r}    \Big\{ \big\| \mathbf{A}(t)   - \mathbf{P}^{\widetilde{\eta}_{k-1}, \widetilde{\eta}_{k}} \big\|_{\mathrm{F}}^2  -  \big\| \mathbf{A}(t) - \mathbf{P}(\eta_k) \big\|_{\mathrm{F}}^2     \Big\} \nonumber\\
 = &  \sum_{t=\eta_k+1}^{\eta_k + r} \big\|  \mathbf{P}(\eta_k)  - \mathbf{P}^{\widetilde{\eta}_{k-1}, \widetilde{\eta}_{k}} \big\|_{\mathrm{F}}^2  -    2 \sum_{t=\eta_k+1}^{\eta_k + r}  \big\langle \mathbf{A}(t)  -   \mathbf{P}(\eta_k), \mathbf{P}^{\widetilde{\eta}_{k-1}, \widetilde{\eta}_{k}} - \mathbf{P}(\eta_k) \big\rangle
 \nonumber\\
 = &  \sum_{t=\eta_k+1}^{\eta_k + r} \big\|  \mathbf{P}(\eta_k)  - \mathbf{P}^{\widetilde{\eta}_{k-1}, \widetilde{\eta}_{k}} \big\|_{\mathrm{F}}^2  -    2 \sum_{t=\eta_k+1}^{\eta_k + r}  \big\langle \mathbf{P}(\eta_{k+1})  -   \mathbf{P}(\eta_k), \mathbf{P}^{\widetilde{\eta}_{k-1}, \widetilde{\eta}_{k}} - \mathbf{P}(\eta_k) \big\rangle \nonumber\\
 & \hspace{0.5cm}  -    2 \sum_{t=\eta_k+1}^{\eta_k + r}  \big\langle \mathbf{A}(t)  -   \mathbf{P}(\eta_{k+1}), \mathbf{P}^{\widetilde{\eta}_{k-1}, \widetilde{\eta}_{k}} - \mathbf{P}(\eta_k) \big\rangle \nonumber\\
 = & III.1 - 2  III.2 -2 III.3. 
\end{align}
From \eqref{eq-inference-1}, we obtain that 
\begin{align}\label{eq-inference-III.1}
  \vert III.1 \vert =  O_p( r \alpha_T^{-2}  \kappa_{k}^2).
\end{align}
Using the Cauchy--Schwarz inequality and again \eqref{eq-inference-1}, we have that
\begin{align}\label{eq-inference-III.2}
    \vert III.2 \vert  \leq   r  \big\| \mathbf{P}(\eta_{k+1})  -   \mathbf{P}(\eta_k)\big\|_{\mathrm{F}} \big\|\mathbf{P}^{\widetilde{\eta}_{k-1}, \widetilde{\eta}_{k}} - \mathbf{P}(\eta_k) \big\|_{\mathrm{F}} = O_p(  r\alpha_T^{-1} \kappa_k^2).
\end{align}
To bound $III.3$, by Lemma~\ref{lemma-s-2} and \eqref{eq-inference-1}, we get that 
\begin{align}\label{eq-inference-III.3}
\vert III.3 \vert  =  O_p \Big( r^{1/2} \big\|\mathbf{P}^{\widetilde{\eta}_{k-1}, \widetilde{\eta}_{k}} - \mathbf{P}(\eta_k) \big\|_{\mathrm{F}}  \Big) =  O_p(r^{1/2} \alpha_T^{-1/2}   \kappa_k ). 
\end{align}
Combining  \eqref{eq-inference-III.0},  \eqref{eq-inference-III.1}, \eqref{eq-inference-III.2} and \eqref{eq-inference-III.3}, we conclude that
\begin{align}\label{eq-inference-III}
\vert III \vert  =   o_p (r \kappa_k^2+  r^{1/2}\kappa_k).
\end{align}

\medskip
\noindent
\textbf{Step 3.4.} Order of magnitude of $IV$. Consider the term
\begin{align}\label{eq-inference-IV.0}
 IV= &  \sum_{t=\eta_k+1}^{\eta_k + r}  \Big\{  \big\| \mathbf{A}(t)  - \mathbf{P}^{\widetilde{\eta}_{k}, \widetilde{\eta}_{k+1}} \big\|_{\mathrm{F}}^2  -  \big\| \mathbf{A}(t) - \mathbf{P}(\eta_{k+1}) \big\|_{\mathrm{F}}^2  \Big\}   \nonumber\\
 = &  \sum_{t=\eta_k+1}^{\eta_k + r} \big\|  \mathbf{P}(\eta_{k+1})  - \mathbf{P}^{\widetilde{\eta}_{k}, \widetilde{\eta}_{k+1}} \big\|_{\mathrm{F}}^2  -    2 \sum_{t=\eta_k+1}^{\eta_k + r}  \big\langle \mathbf{A}(t)  -   \mathbf{P}(\eta_{k+1}), \mathbf{P}^{\widetilde{\eta}_{k}, \widetilde{\eta}_{k+1}} - \mathbf{P}(\eta_{k+1}) \big\rangle \nonumber\\
 = &  IV.1 - 2 IV.2. 
\end{align}
By \eqref{eq-inference-2}, we derive that 
\begin{align}\label{eq-inference-IV.1}
  \vert IV.1 \vert =O_p( r \alpha_T^{-2}  \kappa_{k}^2).
\end{align}
By Lemma~\ref{lemma-s-2} and \eqref{eq-inference-2}, we have that 
\begin{align}\label{eq-inference-IV.2}
\vert IV.2 \vert  =  O_p \Big( r^{1/2} \big\|  \mathbf{P}(\eta_{k+1})  - \mathbf{P}^{\widetilde{\eta}_{k}, \widetilde{\eta}_{k+1}} \big\|_{\mathrm{F}} \Big) =O_p(r^{1/2} \alpha_T^{-1/2}   \kappa_k ). 
\end{align}
Combining \eqref{eq-inference-IV.0}, \eqref{eq-inference-IV.1}  and \eqref{eq-inference-IV.2}, we conclude that
\begin{align}\label{eq-inference-IV}
\vert IV \vert  =    o_p (r \kappa_k^2 +  r^{1/2}\kappa_k).
\end{align}

\medskip
\noindent
\textbf{Step 3.5.} Order of magnitude of $V$. 
We now analyze the final term
\begin{align}\label{eq-inference-V}
 V =  & \sum_{t=\eta_k+1}^{\eta_k + r}  \big\| \mathbf{A}(t)  - \mathbf{P}(\eta_{k})\big\|_{\mathrm{F}}^2  -  \big\| \mathbf{A}(t) - \mathbf{P}(\eta_{k+1}) \big\|_{\mathrm{F}}^2 \nonumber\\
  = &  \sum_{t=\eta_k+1}^{\eta_k + r} \big\|  \mathbf{P}(\eta_k)  -  \mathbf{P}(\eta_{k+1})\big\|_{\mathrm{F}}^2  -    2 \sum_{t=\eta_k+1}^{\eta_k + r}  \big\langle \mathbf{A}(t)  -   \mathbf{P}(\eta_{k+1}), \mathbf{P}(\eta_{k})- \mathbf{P}(\eta_{k+1}) \big\rangle
  \nonumber\\
  = & r \kappa_{k}^2 -2 V.1
\end{align}
Using Lemma~\ref{lemma-s-2}, we obtain that 
\begin{align}\label{eq-inference-V.1}
\vert V.1 \vert  = O_p\Big( r^{1/2}\big\|  \mathbf{P}(\eta_k)  -  \mathbf{P}(\eta_{k+1})\big\|_{\mathrm{F}}  \Big)  =   O_p( r^{1/2}  \kappa_k ). 
\end{align}

\medskip
\noindent
\textbf{Step 3.6: } 
Combining \eqref{eq-theorem-3-main}, \eqref{eq-inference-I}, \eqref{eq-inference-II}, \eqref{eq-inference-III}, \eqref{eq-inference-IV}, \eqref{eq-inference-V} and \eqref{eq-inference-V.1} we have for all $r \kappa_k^2 \geq 1$ that 
\[
r  \kappa_k^2 = O_p (1).
\]

\medskip
\noindent
\textbf{Step 4.} Limiting Distributions. For any  $t \in (\tilde{s}_k, \tilde{e}_k)$, define
\[
\widetilde{\mathcal{Q}}_k(t)      =  \sum_{u=\tilde{s}_k+1}^{t} \| \mathbf{A}(u) - \mathbf{P}(\eta_{k}) \|_{\mathrm{F}}^2 + \sum_{u=t +1}^{\tilde{e}_k} \| \mathbf{A}(u) - \mathbf{P}(\eta_{k+1})\|_{\mathrm{F}}^2.
\]
Note that the term $V$ defined in \eqref{eq-theorem-3-main} satisfies 
\[
    V = \widetilde{\mathcal{Q}}_k(\eta_k+ r) - \widetilde{\mathcal{Q}}_k (\eta_k), 
\]
and hence by \eqref{eq-theorem-3-main}, \eqref{eq-inference-I}, \eqref{eq-inference-II}, \eqref{eq-inference-III}, \eqref{eq-inference-IV} and $r  \kappa_k^2 = O_p (1)$,  we have that 
\[
   \big\vert  \mathcal{Q}_k(\eta_k+ r) - \mathcal{Q}_k (\eta_k)    - \big\{  \widetilde{\mathcal{Q}}_k(\eta_k+ r) - \widetilde{\mathcal{Q}}_k (\eta_k)   \big\} \big\vert  \leq |I| + |II| + |III| + |IV| \overset{p}{\rightarrow} 0.
\]
Therefore, by Slutsky's theorem, it suffices to derive the limiting distributions of $\widetilde{\mathcal{Q}}_k(\eta_k+ r) - \widetilde{\mathcal{Q}}_k (\eta_k)$ as $T \to \infty$. We consider the two scenarios for $\kappa_k$.

\medskip
\noindent
\textbf{Non-vanishing scenario.} Suppose $\kappa_k \to \rho_k$, as $T \to \infty$, with $\rho_k > 0$ being an absolute constant. For $r < 0$, we have that 
\begin{align}
   \widetilde{\mathcal{Q}}_k(\eta_k+ r) - \widetilde{\mathcal{Q}}_k (\eta_k) = & \sum_{t=\eta_k+r+1}^{\eta_k}    \Big\{  \big\| \mathbf{A}(t)  - \mathbf{P}(\eta_{k+1})\big\|_{\mathrm{F}}^2  -  \big\| \mathbf{A}(t) - \mathbf{P}(\eta_{k}) \big\|_{\mathrm{F}}^2  \Big\}  \nonumber\\
    = & \sum_{t=\eta_k+r+1}^{\eta_k} \big\|  \mathbf{P}(\eta_k)  -  \mathbf{P}(\eta_{k+1})\big\|_{\mathrm{F}}^2  
    \nonumber \\ 
    & \hspace{0.5cm} -    2 \sum_{t=\eta_k+1}^{\eta_k + r}  \big\langle \mathbf{A}(t)  -   \mathbf{P}(\eta_{k}),\mathbf{P}(\eta_{k+1}) -  \mathbf{P}(\eta_{k})  \big\rangle \nonumber \\
    \xrightarrow{\mathcal{D}} &  \,
    - r  \rho_k^2 -2 \rho_k \sum_{t=r+1}^{0}  \langle \boldsymbol{\mathbf{\Psi}}_k ,\mathbf{E}_{k}(t) \rangle, 
\end{align} 
with $\boldsymbol{\mathbf{\Psi}}_k$ defined in Model~\ref{model-1}, and  for any $k \in [K+1]$ and $t \in \mathbb{Z}$,
     $\mathbf{E}_{k}(t) = \mathbf{A}_{k}(t) - \mathbf{P}(\eta_{k})$ with $\{\mathbf{A}_{k}(t)\}_{t \in \mathbb{Z}} \overset{\mathrm{i.i.d.}}{\sim}  \mathrm{MRDPG}(\{X_i\}_{i=1}^{n}, \{W_{(l)}(\eta_{k})\}_{l\in [L]})$.
 
For $r>0$, we have that when $T \to \infty$,
\begin{align}
    \widetilde{\mathcal{Q}}_k(\eta_k+ r) - \widetilde{\mathcal{Q}}_k (\eta_k) = & \sum_{t=\eta_k+1}^{\eta_k + r}    \Big\{  \big\| \mathbf{A}(t)  - \mathbf{P}(\eta_{k})\big\|_{\mathrm{F}}^2  -  \big\| \mathbf{A}(t) - \mathbf{P}(\eta_{k+1}) \big\|_{\mathrm{F}}^2  \Big\}  \nonumber\\
    = & \sum_{t=\eta_k+1}^{\eta_k + r} \big\|  \mathbf{P}(\eta_k)  -  \mathbf{P}(\eta_{k+1})\big\|_{\mathrm{F}}^2     
    \nonumber \\ 
    & \hspace{0.5cm} + 2 \sum_{t=\eta_k+1}^{\eta_k + r}  \big\langle \mathbf{A}(t)  -   \mathbf{P}(\eta_{k+1}),\mathbf{P}(\eta_{k+1}) -  \mathbf{P}(\eta_{k})  \big\rangle \nonumber \\
    \xrightarrow{\mathcal{D}}&  \, r \rho_k^2 + 2 \rho_k  \sum_{t=1}^{r} \langle \boldsymbol{\mathbf{\Psi}}_k ,\mathbf{E}_{k+1}(t) \rangle.
\end{align}
By Slutsky’s theorem and the argmin continuous mapping theorem \citep[see e.g.~Theorem 3.2.2 in][]{wellner2013weak}, we obtain 
\[
     \widehat{\eta}_k - \eta_k \xrightarrow{\mathcal{D}} \argmin \mathcal{P}_k(r),
\]
which completes the proof of part Theorem~\ref{theorem-inference-cont}.

\medskip
\noindent
\textbf{Vanishing scenario.} Let $m = \kappa_k^{-2}$, noting that $m \to \infty$ as $T \to \infty$. For $r > 0$, we have that
\begin{align*}
 &  \widetilde{\mathcal{Q}}_k(\eta_k+ rm) - \widetilde{\mathcal{Q}}_k (\eta_k)   \nonumber\\
   =  & \sum_{t=\eta_k+1}^{\eta_k + rm}    \Big\{  \big\| \mathbf{A}(t)  - \mathbf{P}(\eta_{k})\big\|_{\mathrm{F}}^2  -  \big\| \mathbf{A}(t) - \mathbf{P}(\eta_{k+1}) \big\|_{\mathrm{F}}^2  \Big\}   \nonumber\\
   = &  \sum_{t=\eta_k+1}^{\eta_k + rm} \big\|  \mathbf{P}(\eta_k)  -  \mathbf{P}(\eta_{k+1})\big\|_{\mathrm{F}}^2  +    2 \sum_{t=\eta_k+1}^{\eta_k + rm}  \big\langle \mathbf{A}(t)  -   \mathbf{P}(\eta_{k+1}),\mathbf{P}(\eta_{k+1}) -  \mathbf{P}(\eta_{k})  \big\rangle \nonumber \\
    =  & r + \frac{2}{\sqrt{m}}  \sum_{t=\eta_k+1}^{\eta_k + rm}  \big\langle \mathbf{A}(t)  -   \mathbf{P}(\eta_{k+1}), \boldsymbol{\mathbf{\Psi}}_k \big\rangle. 
\end{align*}
By the functional central limit theorem, we have that when $T \to \infty$,
\[
\frac{1}{\sqrt{m}} \sum_{t = \eta_k+ 1}^{\eta_k + rm } \big\langle \mathbf{A}(t)  -   \mathbf{P}(\eta_{k+1}), \boldsymbol{\mathbf{\Psi}}_k  \big\rangle  \xrightarrow{\mathcal{D}}    \sigma_{k, k+1} \mathbb{B}_1(r),
\]
where $\mathbb{B}_1(r)$ is a standard Brownian motion and  for any $k \in [K]$ and $k' \in \{k, k+1\}$,
$\sigma_{k, k'}^2 =  \operatorname{Var} \big( \langle \boldsymbol{\mathbf{\Psi}}_k ,\mathbf{E}_{k'}(1) \rangle \big)$.  
Consequently, as~$T \to \infty$
\[
\widetilde{\mathcal{Q}}_k(\eta_k+ rm) - \widetilde{\mathcal{Q}}_k (\eta_k) 
 \xrightarrow{\mathcal{D}}  r +2  \sigma_{k, k+1} \mathbb{B}_1(r).
\]
Similarly, for $r < 0$, we have that when $T \to \infty$
\[
\widetilde{\mathcal{Q}}_k(\eta_k+ rm) - \widetilde{\mathcal{Q}}_k (\eta_k) 
\overset{D}{\to} - r + 2 \sigma_{k, k} \mathbb{B}_2(-r),
\]
where $\mathbb{B}_2(r)$ is a standard Brownian motion.
Applying Slutsky’s theorem and the argmin continuous mapping theorem \citep[see e.g.~Theorem 3.2.2 in][]{wellner2013weak}, we conclude that 
\[
     \kappa_k^2 (\widehat{\eta}_k - \eta_k)\xrightarrow{\mathcal{D}} \argmin \mathcal{P}_k'(r),
\]
which completes the proof of Theorem~\ref{theorem-inference}.

\end{proof}

\section{Additional details and results in Section \ref{sec:experiment}}\label{sec-add-simulation}

All experiments were run on a CPU with 16GB RAM. For each synthetic scenario with node size $n=100$, number of layers $L=4$ and time span $T=200$, the compute time is about $10$ hours to localize the change points and to construct the confidence intervals over $100$ Monte Carlo trials. For each real data experiment, the computation time is approximately $15$ minutes to perform change point localization and confidence interval construction.

\subsection{Additional results in Section \ref{sec:simulation}}\label{appendix_syn_data}

Table \ref{tbl:s3s4} presents the results for \textbf{Scenarios 3} and \textbf{4}.

\begin{table}[!htb]
\caption{Means of evaluation metrics for networks simulated from Scenarios 3 and 4.}
\vspace{1em}
\label{tbl:s3s4}
\centering
\resizebox{\textwidth}{!}{ 
\begin{tabular}{ ll llll  llll}
\toprule
&  & \multicolumn{4}{c}{Scenario 3} & \multicolumn{4}{c}{Scenario 4} \\
$n$ & Method & $|\widehat{K}-K|\downarrow$ & $d(\widehat{\mathcal{C}}|\mathcal{C})\downarrow$ & $d(\mathcal{C}|\widehat{\mathcal{C}})\downarrow$ & $C(\mathcal{G},\mathcal{G'})\uparrow$ & $|\widehat{K}-K|\downarrow$ & $d(\widehat{\mathcal{C}}|\mathcal{C})\downarrow$ & $d(\mathcal{C}|\widehat{\mathcal{C}})\downarrow$ & $C(\mathcal{G},\mathcal{G'})\uparrow$\\ 
\midrule
\multirow{5}{2em}{$50$} 
& CPDmrdpg       & $0.19$ & $9.64$ & $0.14$ & $95.11\%$                & $0.00$ & $0.02$ & $0.02$ & $99.98\%$ \\
& gSeg (nets.)   & $0.98$ & $\text{Inf}$ & $\text{Inf}$ & $68.93\%$    & $5.00$ & $\text{Inf}$ & $\text{Inf}$ & $0.00\%$ \\
& kerSeg (nets.) & $0.16$ & $0.18$ & $2.06$ & $98.90\%$                & $0.36$ & $0.14$ & $2.65$ & $98.56\%$ \\
& gSeg (frob.)   & $0.92$ & $\text{Inf}$ & $\text{Inf}$ & $66.78\%$    & $1.53$ & $\text{Inf}$ & $\text{Inf}$ & $74.92\%$ \\
& kerSeg (frob.) & $0.82$ & $48.52$ & $5.11$ & $73.55\%$               & $0.40$ & $0.05$ & $3.71$ & $98.12\%$ \\
\midrule
\multirow{5}{2em}{$100$} 
& CPDmrdpg       & $0.00$ & $0.02$ & $0.02$ & $99.98\%$                & $0.00$ & $0.00$ & $0.00$ & $100\%$ \\
& gSeg (nets.)   & $0.69$ & $\text{Inf}$ & $\text{Inf}$ & $80.10\%$    & $4.98$ & $\text{Inf}$ & $\text{Inf}$ & $0.77\%$ \\
& kerSeg (nets.) & $0.17$ & $0.00$ & $3.26$ & $99.16\%$                & $0.34 $& $0.08$ & $2.93$ & $98.47\%$ \\
& gSeg (frob.)   & $0.79$ & $\text{Inf}$ & $\text{Inf}$ & $72.11\%$    & $1.86$ & $\text{Inf}$ & $\text{Inf}$ & $68.57\%$  \\
& kerSeg (frob.) & $0.79$ & $48.82$ & $4.75$ & $73.80\%$               & $0.42$ & $0.06$ & $2.93$ & $98.63\%$ \\
\bottomrule
\end{tabular}
} 
\end{table}

We then present a sensitivity analysis of the threshold constant $c_{\tau,1}$ used in Algorithm~\ref{offline-algorithm}. Specifically, based on Theorem~\ref{theorem-2}, we set the threshold value $\tau=  c_{\tau,1} n\sqrt{L} \log^{3/2}(T)$. The constant $c_{\tau,1}$ was chosen empirically assessing the false positive rate in an MSBM model without any change points. The construction follows the details provided in Section~\ref{sec:simulation}, with four evenly-sized communities. We found that $c_{\tau,1} = 0.1$ detects a change point approximately $1\%$ of the time in this scenario, demonstrating ideal false positive control. Smaller values of $c_{\tau,1}$, such as $c_{\tau,1} \in \{0.05, 0.08\}$, led to more frequent false detections, while larger values $c_{\tau,1} \in \{0.12, 0.15, 0.20\}$ failed to detect change points in this scenario.

Tables \ref{tbl:sensitivity_s1}–\ref{tbl:sensitivity_s4} report the results of the sensitivity analysis for \textbf{Scenarios 1} through \textbf{4}, varying $c_{\tau,1} \in \{0.05, 0.08, 0.10, 0.12, 0.15, 0.20, 0.25\}$. These results demonstrate that our proposed method is relatively robust against the choices of $c_{\tau,1}$.

\begin{table}[!htb]
\caption{Means of evaluation metrics for dynamic networks simulated from Scenario 1, varying $c_{\tau,1}$.}
\vspace{0.4em}
\label{tbl:sensitivity_s1}
\begin{center}
\resizebox{0.7\textwidth}{!}{ 
\begin{tabular}{ llp{1.7cm}p{1.7cm}p{1.7cm}p{1.7cm} }
\toprule
$n$ & $c_{\tau,1}$ & $|\widehat{K}-K|\downarrow$ & $d(\widehat{\mathcal{C}}|\mathcal{C})\downarrow$ & $d(\mathcal{C}|\widehat{\mathcal{C}})\downarrow$ & $C(\mathcal{G},\mathcal{G'})\uparrow$\\
\midrule
\multirow{7}{2em}{$50$}
& $0.25$  & $0.00$    & $0.00$  & $0.00$ & $100\%$ \\
& $0.20$  & $0.00$    & $0.00$  & $0.00$ & $100\%$ \\
& $0.15$  & $0.00$    & $0.00$  & $0.00$ & $100\%$ \\
& $0.12$  & $0.00$    & $0.00$  & $0.00$ & $100\%$ \\
& $0.10$  & $0.01$    & $0.00$  & $0.42$ & $99.86\%$ \\
& $0.08$  & $0.25$    & $0.00$  & $6.68$ & $97.80\%$ \\
& $0.05$  & $5.18$    & $0.00$  & $52.86$ & $67.50\%$ \\
\midrule
\multirow{7}{2em}{$100$}
& $0.25$ & $0.00$    & $0.00$  & $0.00$ & $100\%$ \\
& $0.20$ & $0.00$    & $0.00$  & $0.00$ & $100\%$ \\
& $0.15$ & $0.00$    & $0.00$  & $0.00$ & $100\%$ \\
& $0.12$ & $0.00$    & $0.00$  & $0.00$ & $100\%$ \\
& $0.10$ & $0.00$    & $0.00$  & $0.00$ & $100\%$ \\
& $0.08$ & $0.15$    & $0.00$  & $4.98$ & $98.54\%$ \\
& $0.05$ & $5.02$    & $0.00$  & $53.84$ & $67.56\%$ \\
\bottomrule
\end{tabular}
}
\end{center}
\end{table}

\begin{table}[!htb]
\caption{Means of evaluation metrics for dynamic networks simulated from Scenario 2, varying $c_{\tau,1}$.}
\vspace{0.4em}
\label{tbl:sensitivity_s2}
\begin{center}
\resizebox{0.7\textwidth}{!}{ 
\begin{tabular}{ llp{1.7cm}p{1.7cm}p{1.7cm}p{1.7cm} } 
\toprule
$n$ & $c_{\tau,1}$ & $|\widehat{K}-K|\downarrow$ & $d(\widehat{\mathcal{C}}|\mathcal{C})\downarrow$ & $d(\mathcal{C}|\widehat{\mathcal{C}})\downarrow$ & $C(\mathcal{G},\mathcal{G'})\uparrow$\\
\midrule
\multirow{7}{2em}{$50$}
& $0.25$  & $0.00$    & $0.00$  & $0.00$ & $100\%$ \\
& $0.20$  & $0.00$    & $0.00$  & $0.00$ & $100\%$ \\
& $0.15$  & $0.00$    & $0.00$  & $0.00$ & $100\%$ \\
& $0.12$  & $0.00$    & $0.00$  & $0.00$ & $100\%$ \\
& $0.10$  & $0.00$    & $0.00$  & $0.00$ & $100\%$ \\
& $0.08$  & $0.02$    & $0.00$  & $0.64$ & $99.68\%$ \\ 
& $0.05$  & $3.79$    & $0.00$  & $28.46$ & $75.43\%$ \\
\midrule
\multirow{7}{2em}{$100$}
& $0.25$ & $0.00$    & $0.00$  & $0.00$ & $100\%$ \\
& $0.20$ & $0.00$    & $0.00$  & $0.00$ & $100\%$ \\
& $0.15$ & $0.00$    & $0.00$  & $0.00$ & $100\%$ \\
& $0.12$ & $0.00$    & $0.00$  & $0.00$ & $100\%$ \\
& $0.10$ & $0.00$    & $0.00$  & $0.00$ & $100\%$ \\
& $0.08$ & $0.05$    & $0.00$  & $1.14$ & $99.38\%$ \\
& $0.05$ & $3.53$    & $0.00$  & $28.60$ & $76.50\%$ \\ 
\bottomrule
\end{tabular}
}
\end{center}
\end{table}

\begin{table}[!htb]
\caption{Means of evaluation metrics for dynamic networks simulated from Scenario 3, varying $c_{\tau,1}$.}
\vspace{0.4em}
\label{tbl:sensitivity_s3}
\begin{center}
\resizebox{0.7\textwidth}{!}{ 
\begin{tabular}{ llp{1.7cm}p{1.7cm}p{1.7cm}p{1.7cm} }
\toprule
$n$ & $c_{\tau,1}$ & $|\widehat{K}-K|\downarrow$ & $d(\widehat{\mathcal{C}}|\mathcal{C})\downarrow$ & $d(\mathcal{C}|\widehat{\mathcal{C}})\downarrow$ & $C(\mathcal{G},\mathcal{G'})\uparrow$\\
\midrule
\multirow{7}{2em}{$50$}
& $0.25$ & $1.00$ & $50.00$  & $0.00$ & $75.00\%$ \\
& $0.20$ & $0.96$ & $48.00$  & $0.00$ & $76.00\%$ \\
& $0.15$ & $0.64$ & $32.00$  & $0.00$ & $84.00\%$ \\
& $0.12$ & $0.39$ & $19.58$  & $0.08$ & $90.17\%$ \\
& $0.10$ & $0.19$ & $ 9.64$  & $0.14$ & $95.11\%$ \\
& $0.08$ & $0.09$ & $ 4.30$  & $0.50$ & $97.61\%$ \\
& $0.05$ & $4.27$ & $0.36$   & $32.54$ & $71.55\%$ \\
\midrule
\multirow{7}{2em}{$100$}
& $0.25$ & $0.43$  & $21.50$ & $0.00$ & $89.25\%$ \\
& $0.20$ & $0.15$  & $7.52$  & $0.02$ & $96.23\%$ \\
& $0.15$ & $0.00$  & $0.02$  & $0.02$ & $99.98\%$ \\
& $0.12$ & $0.00$  & $0.02$  & $0.02$ & $99.98\%$ \\
& $0.10$ & $0.00$  & $0.02$  & $0.02$ & $99.98\%$ \\
& $0.08$ & $0.06$  & $0.02$  & $1.08$ & $99.57\%$ \\
& $0.05$ & $4.04$  & $0.02$  & $32.76$ & $73.62\%$ \\
\bottomrule
\end{tabular}
}
\end{center}
\end{table}

\begin{table}[!htb]
\caption{Means of evaluation metrics for dynamic networks simulated from Scenario 4, varying $c_{\tau,1}$.}
\vspace{0.4em}
\label{tbl:sensitivity_s4}
\begin{center}
\resizebox{0.7\textwidth}{!}{
\begin{tabular}{ llp{1.7cm}p{1.7cm}p{1.7cm}p{1.7cm} } 
\toprule
$n$ & $c_{\tau,1}$ & $|\widehat{K}-K|\downarrow$ & $d(\widehat{\mathcal{C}}|\mathcal{C})\downarrow$ & $d(\mathcal{C}|\widehat{\mathcal{C}})\downarrow$ & $C(\mathcal{G},\mathcal{G'})\uparrow$\\
\midrule
\multirow{7}{2em}{$50$}
& $0.25$ & $2.67$  & $83.20$  & $0.00$ & $62.47\%$ \\
& $0.20$ & $1.19$  & $28.40$  & $0.00$ & $85.63\%$ \\
& $0.15$ & $0.13$  & $2.60$   & $0.00$ & $98.67\%$ \\
& $0.12$ & $0.01$  & $0.22$   & $0.02$ & $99.88\%$ \\
& $0.10$ & $0.00$  & $0.02$   & $0.02$ & $99.98\%$ \\
& $0.08$ & $0.00$  & $0.02$   & $0.02$ & $99.98\%$ \\
& $0.05$ & $0.75$  & $0.02$   & $11.94$ & $93.36\%$ \\
\midrule
\multirow{7}{2em}{$100$}
& $0.25$ & $0.01$  & $0.20$  & $0.00$ & $99.90\%$ \\
& $0.20$ & $0.00$  & $0.00$  & $0.00$ & $100\%$ \\
& $0.15$ & $0.00$  & $0.00$  & $0.00$ & $100\%$ \\
& $0.12$ & $0.00$  & $0.00$  & $0.00$ & $100\%$ \\
& $0.10$ & $0.00$  & $0.00$  & $0.00$ & $100\%$ \\
& $0.08$ & $0.00$  & $0.00$  & $0.00$ & $100\%$ \\
& $0.05$ & $0.89$  & $0.00$  & $12.46$ & $92.55\%$ \\
\bottomrule
\end{tabular}
}
\end{center}
\end{table}

\subsection{Additional details and results in Section \ref{sec:real}}\label{appendix_real_data}

This section provides a detailed analysis of the U.S.~air transportation network data, evaluates the performance of competing methods (introduced in Section~\ref{sec:simulation}) on both real datasets and presents the constructed confidence intervals using the procedure in Section~\ref{sec:CI}. 

\noindent\textbf{The U.S.~air transportation network data} consist of monthly data from January 2015 to June 2022 ($T=90$) and are available from \cite{BTS2022}. Each node corresponds to an airport and each layer represents a commercial airline. A directed edge in a given layer indicates a direct flight operated by a specific commercial airline between two airports. We choose the $L=4$ airlines with the highest flight volumes and the $n=50$ airports with the most departures and arrivals. Our method identifies change points in December 2015, June 2017, February 2019, February 2020 and February 2021, all corresponding to major abruptions in the U.S.~aviation industry. 

The change point in December 2015 coincides with increased regulatory scrutiny over airline consolidation, following concerns raised by the American Antitrust Institute about reduced market competition after a series of mergers. The June 2017 change point aligns with the proposal of the Aviation Innovation, Reform and Reauthorization Act, which advocated for privatizing air traffic control and influenced route planning among carriers. Moreover, the February 2019 change point follows the U.S.~government shutdown (December 2018 - January 2019), which caused Transportation Security Administration staffing shortages and significant operational disruptions, prompting stabilization efforts in the months that followed.  Lastly, the most significant structural disruptions emerged in February 2020 and February 2021, aligning with the initial shock and continued fallout of the COVID-19 pandemic, which triggered widespread flight cancellations, demand collapse and structural reconfiguration in the aviation industry.

\noindent \textbf{Performance of competitors.}
Table \ref{tbl:trade} summarizes the change points detected by the proposed and competing methods
for the worldwide agricultural trade network data.  Notably, the gSeg method fails to detect any change points after 2010, regardless of input type. Meanwhile, the kerSeg method detects change points in 1990 and 1992, which are temporally too close. In contrast, our proposed method (CPDmrdpg) identifies four major change points that align well with known geopolitical and policy-related events.

Table \ref{tbl:AirTransp} presents the results for the U.S.~air transportation network data. Although the kerSeg method using networks as input demonstrates a good performance in the simulation study, it detects an excessive number of change points in this real data experiment, making the results unreliable and raising concerns about false positives. Similarly, the kerSeg method that uses layer-wise Frobenius norms as input has detected change points that are too close, yielding clusters of change points that could potentially be grouped together. On the contrary, the gSeg method that uses the Frobenius norms as input detects too few change points, while the gSeg method using networks as input has detected too many change points. The proposed CPDmrdpg method (Algorithm~\ref{offline-algorithm}) yields five change points that align well with known disruptions and policy changes in the aviation sector.

While the competitor methods do detect important and relevant change points in both two real datasets, they tend to either over- or under-segment the time span. These patterns suggest that the change points identified by the competing methods may be less realistic or informative compared to those identified by the proposed method.

\begin{table}[!htb]
\caption{Detected change points for the worldwide agricultural trade network data.}
\vspace{0.5em}
\begin{center}
\label{tbl:trade}
\resizebox{0.5\textwidth}{!}{
\begin{tabular}{ ll }  
\toprule
Method & Detected change points\\
\midrule
CPDmrdpg       & 1991, 1999, 2005, 2013\\
gSeg (nets.)   & 1993, 2002, 2010\\
kerSeg (nets.) & 1990, 1992, 1999, 2005, 2012\\
gSeg (frob.)   & 1993, 2002, 2009\\
kerSeg (frob.) & 1990, 1992, 1997, 2003, 2012\\
\bottomrule
\end{tabular}
}
\end{center}
\end{table}

\begin{table}[!htb]
\caption{Detected change points for the U.S.~air transportation network data.}
\vspace{0.4em}
\label{tbl:AirTransp}
\begin{center}
\resizebox{0.8\textwidth}{!}{ 
\begin{tabular}{ ll }  
\toprule
Method & Detected change points\\
\midrule
CPDmrdpg       & 2015-12, 2017-06, 2019-02, 2020-02, 2021-02\\
gSeg (nets.)   & 2015-11, 2016-10, 2017-09, 2018-09, 2019-09, 2020-10, 2021-08\\
kerSeg (nets.) & 2015-11, 2016-03, 2016-10, 2017-05, 2017-09, 2018-05, 2018-10\\
& 2019-03, 2019-09, 2020-03, 2020-10, 2021-03, 2021-09\\
gSeg (frob.)   & 2015-11, 2020-01, 2021-03\\
kerSeg (frob.) & 2015-11, 2017-10, 2020-01, 2021-03, 2021-05, 2021-09, 2022-01\\
\bottomrule
\end{tabular}
}
\end{center}
\end{table}

\noindent \textbf{Performance of constructed confidence intervals.}
Tables~\ref{tbl:trade_CI} and \ref{tbl:AirTransp_CI} report the detected change point from Algorithm~\ref{offline-algorithm} and the $95\%$ confidence intervals constructed via the procedure from Section~\ref{sec:CI}, for the agricultural trade and air transportation networks, respectively.

\begin{table}[!htb]
\caption{Detected change point from Algorithm~\ref{offline-algorithm} and $95\%$ confidence intervals via Section~\ref{sec:CI} for the worldwide agricultural trade network data.}
\vspace{0.5em}
\label{tbl:trade_CI}
\begin{center}
\resizebox{0.6\textwidth}{!}{ 
\begin{tabular}{ lll }  
\toprule
Detected change points & Time point & Confidence interval\\
\midrule
1991 & $6$  & $(5.97,  6.03)$\\
1999 & $14$ & $(13.98, 14.02)$\\
2005 & $20$ & $(17.97, 18.05)$\\
2013 & $28$ & $(25.99, 26.06)$\\
\bottomrule
\end{tabular}
}
\end{center}
\end{table}

\begin{table}[!htb]
\caption{Detected change point from Algorithm \ref{offline-algorithm} and $95\%$ confidence intervals via Section~\ref{sec:CI}  for the U.S.~air transportation network data.}
\vspace{0.5em}
\label{tbl:AirTransp_CI}
\begin{center}
\resizebox{0.6\textwidth}{!}{
\begin{tabular}{ lll }  
\toprule
Detected change points & Time point & Confidence interval\\
\midrule
2015-12 & $12$ & $(11.55, 12.41)$\\
2017-06 & $30$ & $(28.79, 30.98)$\\
2019-02 & $50$ & $(49.67, 53.22)$\\
2020-02 & $62$ & $(59.66, 60.36)$\\
2021-02 & $74$ & $(73.58, 74.27)$\\
\bottomrule
\end{tabular}
}
\end{center}
\end{table}

\end{document}